\documentclass[final,runningheads,envcountsect,envcountsame]{llncs}
\pdfoutput=1
\usepackage{hyperref}
\usepackage[draft]{graphicx}
\usepackage[english]{babel}
\usepackage{stmaryrd}
\usepackage{mathtools}
\usepackage{amssymb}
\usepackage[misc]{ifsym}
\usepackage{dsfont}
\usepackage{xspace}
\usepackage{tikz}
\usetikzlibrary{cd,fit}
\usetikzlibrary{calc}
\usepackage{etex}
\usepackage{booktabs}
\usepackage{seqsplit}
\usepackage{xstring}
\usepackage{ifdraft}
\usepackage{environ}
\usepackage{etoolbox}
\usepackage{letltxmacro}
\usepackage{pdftexcmds}
\usepackage{rotating}
\usepackage{wrapfig}
\usepackage{todonotes}
\usepackage{bbold} 
\usepackage{multirow}

\hypersetup{
  final,
  bookmarks=true,
  colorlinks=true,
  linkcolor=black,
  citecolor=black,
  filecolor=black,
  urlcolor=black,
  pdftitle={Path category for free},
  pdfauthor={Thorsten Wißmann, J辿r辿my Dubut, Shin-ya Katsumata, Ichiro Hasuo},
  pdfduplex={DuplexFlipLongEdge},
}

\makeatletter
\newcommand{\customlabel}[2]{%
\protected@write \@auxout {}{\string \newlabel {#1}{{#2}{\thepage}{#2}{#1}{}} }%
\hypertarget{#1}{#2}%
}
\makeatother
\newcounter{axiomcounter}
\makeatletter
\newcommand{\axlabel}[1]{%
\addtocounter{axiomcounter}{1}%
(\customlabel{#1}{Ax\arabic{axiomcounter}})%
}
\makeatother

\makeatletter
\newenvironment{notheorembrackets}{%
\csdef{@spopargbegintheorem}##1##2##3##4##5{\trivlist%
      \item[\hskip\labelsep{##4##1\ ##2}]{##4{##3}\@thmcounterend\ }##5}%
    }{%
\csdef{@spopargbegintheorem}##1##2##3##4##5{\trivlist%
      \item[\hskip\labelsep{##4##1\ ##2}]{##4(##3)\@thmcounterend\ }##5}%
    }
\makeatother

\ifdraft{}{
\usepackage[marginparwidth=2cm,nohead,nofoot,
            headsep = 1cm,
            footskip = 1cm,
            text={12.2cm,19.3cm},
            paperwidth=16.2cm,
            paperheight=23.3cm,
            marginparsep=1mm,
            marginparwidth=1.9cm,
            footskip=1cm,
            centering
            ]{geometry}
}{}

\ifdraft{
\makeatletter
\usepackage[notcite,notref]{showkeys}%

\newcommand\enableshowkeys{\let\hideNextShowKeysLabel\undefined}
\renewcommand*\showkeyslabelformat[1]{%
\ifcsname hideNextShowKeysLabel\endcsname%
\else%
\StrSubstitute{#1}{ }{\textvisiblespace}[\TEMP]%
\parbox[t]{\marginparwidth}{\raggedright\normalfont\small\ttfamily\(\{\){\color{red!50!black}\expandafter\seqsplit\expandafter{\TEMP}}\(\}\)}%
\fi%
}
\makeatother
}{}

\newcommand{\A}{\ensuremath{\mathbb{A}\xspace}}
\newcommand{\perms}{\ensuremath{\mathfrak{S}_{\mathsf{f}}\xspace}}
\newcommand{\B}{\ensuremath{\mathcal{B}\xspace}}
\newcommand{\C}{\ensuremath{\mathbb{C}\xspace}}
\newcommand{\PP}{\ensuremath{\mathcal{P}\xspace}}

\newcommand{\E}{\ensuremath{\mathcal{E}\xspace}}

\newcommand{\N}{\ensuremath{\mathds{N}\xspace}}

\newcommand{\fpair}[1]{\ensuremath{\langle #1 \rangle}}
\newcommand{\obj}{\ensuremath{\mathbf{obj}\,\xspace}}
\newcommand{\mor}{\ensuremath{\mathbf{mor}\,\xspace}}
\newcommand{\Id}{\ensuremath{\operatorname{\textnormal{Id}}\xspace}}
\newcommand{\oper}[1]{\ensuremath{\operatorname{\textnormal{\textsf{#1}}}\xspace}}
\newcommand{\id}{\ensuremath{\operatorname{\textnormal{id}}\xspace}}
\newcommand{\Set}{\oper{Set}\xspace}
\newcommand{\Pos}{\oper{Pos}\xspace}

\newcommand{\Nom}{\oper{Nom}}

\newcommand{\supp}{\oper{supp}}
\newcommand{\dom}{\oper{dom}}
\newcommand{\inl}{\ensuremath{\mathsf{inl}\xspace}}
\newcommand{\inr}{\ensuremath{\mathsf{inr}\xspace}}
\newcommand{\inj}{\ensuremath{\mathsf{in}\xspace}}
\newcommand{\ar}{\ensuremath{\mathsf{ar}\xspace}}
\newcommand{\pow}{\ensuremath{\mathcal{P}\xspace}}
\newcommand{\powf}{\ensuremath{\mathcal{P}_{\mathsf{f}}\xspace}}
\newcommand{\powufs}{\ensuremath{\mathcal{P}_{\mathsf{ufs}}\xspace}}
\newcommand{\Bag}{\ensuremath{\mathcal{B}\xspace}}

\newcommand{\downinclusion}{\begin{turn}{270}{$\sqsubseteq$}\end{turn}}
\newcommand{\upinclusion}{\begin{turn}{90}{$\sqsubseteq$}\end{turn}}
\renewcommand{\S}{\ensuremath{\mathcal{S}\xspace}}

\newcommand{\Path}{\ensuremath{\text{\upshape\sf Path}}\xspace}

\newcommand{\trace}{\ensuremath{\text{\upshape\sf tr}}\xspace}
\newcommand{\PathOrd}{\ensuremath{\mathsf{PathOrd}}\xspace}
\newcommand{\Comp}{\ensuremath{\text{\upshape\sf Comp}}\xspace}

\newcommand{\M}{\ensuremath{\mathbb{M}}}
\renewcommand{\arg}{\ensuremath{\_\!\_}}

\newcommand{\LCoalg}{\ensuremath{\operatorname{\sf Coalg}_l}}

\newcommand{\bbP}{\ensuremath{\mathbb{P}}}

\newcommand{\map}[3]{#1\colon #2 \longrightarrow #3}
\renewcommand{\to}{\mathbin{\longrightarrow}}
\newcommand{\trsy}{\oper{LTS}_{A}\xspace} 

\newcommand{\trans}[3]{#1\xrightarrow{~#2}#3}
\newcommand{\barbind}[1]{{|}_{#1}}

\spnewtheorem*{notation}{Notation}{\itshape}{\rmfamily}
\spnewtheorem{assumption}[theorem]{Assumption}{\bfseries}{\rmfamily}
\spnewtheorem{conclusion}[theorem]{Conclusion}{\bfseries}{\rmfamily}

\newsavebox{\mypullbackcorner}%
\sbox{\mypullbackcorner}{%
\begin{tikzpicture}
    \draw[-] (0,0) -- (.5em,.5em) -- (0,1em);
\end{tikzpicture}%
}
\newcommand{\pullbackangle}[2][]{\arrow[phantom,to path={
                     -- ($ (\tikztostart)!1cm!#2:([xshift=8cm]\tikztostart) $)
                        node[anchor=west,pos=0.0,rotate=#2,
                        inner xsep = 0]
                        {\begin{tikzpicture}[minimum
                        height=1mm,baseline=0,#1]
    \draw[-] (0,0) -- (.5em,.5em) -- (0,1em);
                        \end{tikzpicture}}}]{}}
\newcommand{\descto}[3][]{\arrow[phantom]{#2}[#1]{\text{\footnotesize{}#3}}}

\tikzset{
      vertex/.style={
        fill=black,
        shape=circle,
        outer sep = 1mm,
        inner sep = 0mm,
        minimum size=2mm,
      },
      named vertex/.style={
        inner sep = 0mm,
        outer sep = 1mm,
      },
      set/.style={
        draw=black!20,
        fill=none,
        line width=1mm,
        rounded corners=2mm,
        inner sep = 2mm,
      },
      myedge/.style={
        ->,
        draw=black,
        every node/.append style={
          shape=circle,
          inner sep=1pt,
          fill=black,
          text=white,
          draw=none,
          sloped,
          minimum width=5mm,
        },
      },
}

\tikzstyle{shiftarr}=[
        rounded corners,%
        to path={--([#1]\tikztostart.center)
                     -- ([#1]\tikztotarget.center) \tikztonodes
                     -- (\tikztotarget)},
]
\usetikzlibrary{external}
\ifthenelse{\pdfshellescape=1}{
\tikzexternalize[disable dependency files]
}{}
\tikzset{external/up to date check=diff}
\pgfkeys{/pgf/images/include external/.code={\includegraphics[draft=false]{#1}}}

\tikzsetexternalprefix{diagrams/}
\makeatletter
\tikzifexternalizing{%
\def\@enddocumenthook{}%
}{}
\makeatother

\presetkeys{todonotes}{
  color=yellow!20,
  bordercolor=yellow!80!black,
  linecolor=yellow,
  size=\scriptsize,
}{}

\LetLtxMacro{\oldtodo}{\todo}
\renewcommand{\todo}[2][]{\tikzexternaldisable\oldtodo[#1]{#2}\tikzexternalenable}

\def\temp{&} \catcode`&=\active \let&=\temp

\tikzcdset{
arrow style=tikz,
diagrams={>={Straight Barb[scale=0.8]}}
}

\newcommand{\mytikzcdcontext}[2]{
  \begin{tikzpicture}[baseline=(mainnode.base)]
    \node (mainnode) [inner sep=0, outer sep=0] {\begin{tikzcdold}[#2]
        #1
    \end{tikzcdold}};
  \end{tikzpicture}}

\NewEnviron{tikzcd}[1][]{%
\def\myargs{#1}%
\edef\mydiagram{\noexpand\mytikzcdcontext{\expandonce\BODY}{\expandonce\myargs}}%
\mydiagram%
}

\newcommand{\sk}[1]{\textcolor{blue}{\bf Shin-ya: #1}}
\newcommand{\jd}[1]{\textcolor{green!70!black}{\xspace\bf Jeremy: #1}}
\newcommand{\tw}[1]{\textcolor{orange!70!black}{\xspace\bf Thorsten: #1}}

\ifdraft{}{
\def\sk#1{}
\def\jd#1{}
\def\tw#1{}
}

\begin{document}
\title{Path Category For Free\thanks{This research was supported by ERATO HASUO
    Metamathematics for Systems Design Project (No. JPMJER1603), JST. The
    first author was supported by the DFG project MI~717/5-1. He expresses his
    gratitude for having been invited to Tokyo, which initiated the present work.}\\ Open Morphisms
  From Coalgebras With Non-Deterministic Branching}
\titlerunning{Path Category For Free}
%
\author{Thorsten Wi{\ss}mann\inst{1}\orcidID{0000-0001-8993-6486} 
\and
J\'er\'emy Dubut\inst{2,3}
\and\\
Shin-ya Katsumata\inst{2}
\and
Ichiro Hasuo\inst{2,4}
}
\authorrunning{T.~Wi{\ss}mann, J.~Dubut, S.~Katsumata, I.~Hasuo}
%
\institute{Friedrich-Alexander-Universit\"at Erlangen-N\"urnberg, Germany
  \email{thorsten.wissmann@fau.de} \and
National Institute of Informatics, Tokyo, Japan\\
\email{\{dubut,s-katsumata,hasuo\}@nii.ac.jp} \and
Japanese-French Laboratory for Informatics, Tokyo, Japan \and
SOKENDAI, Kanagawa, Japan}
\maketitle              
\begin{abstract}
  There are different categorical approaches to variations of transition systems
  and their bisimulations. One is coalgebra for a functor G, where a
  bisimulation is defined as a span of G-coalgebra homomorphism. Another one is
  in terms of path categories and open morphisms, where a bisimulation is
  defined as a span of open morphisms. This similarity is no coincidence: given
  a functor G, fulfilling certain conditions, we derive a path-category for
  pointed G-coalgebras and lax homomorphisms, such that the open morphisms turn
  out to be precisely the G-coalgebra homomorphisms. The above construction
  provides path-categories and trace semantics for free for different flavours
  of transition systems: (1) non-deterministic tree automata (2) regular
  nondeterministic nominal automata (RNNA), an expressive automata notion living
  in nominal sets (3) multisorted transition systems. This last instance relates
  to Lasota's construction, which is in the converse direction.

\end{abstract}
\section{Introduction}

\emph{Coalgebras} \cite{rutten00} and \emph{open maps} \cite{joyal96} are two
main categorical approaches to transition systems and bisimulations. The former
describes the branching type of systems as an endofunctor, a system becoming a
coalgebra and bisimulations being spans of coalgebra homomorphisms. Coalgebra
theory makes it easy to consider state space types in different settings,
e.g.~nominal sets~\cite{kozen15,kurz13} or algebraic
categories~\cite{bonchi17,hansen11,milius10}. The latter, open maps, describes
systems as objects of a category and the execution types as particular objects
called paths. In this case, bisimulations are spans of open morphisms. Open maps
are particularly adapted to extend bisimilarity to history dependent behaviors,
e.g.~true concurrency \cite{fahrenberg13,dubut15}, timed systems
\cite{nielsen99} and weak (bi)similarity \cite{fiore99}. Coalgebra homomorphisms
and open maps are then key concepts to describe bisimilarity categorically. They
intuitively correspond to functional bisimulations, that is, those maps between
states whose graph is a bisimulation.

\begin{table}[tbp] \centering
  \def\arraystretch{1.1}
  \vspace{-2mm}
  \hspace{-15mm} 
  \begin{tabular}{rccccccl}
  \cmidrule[\heavyrulewidth]{2-7}
  & worlds & data & systems & func. sim. & func. bisim.\, & (bi)simulation& \\
  \cmidrule{2-7}
  \multirow{4}{*}{\rotatebox{90}{$\xrightarrow{\text{this paper}}$}}& open & $\map{J}{\bbP}{\M}$ & \multirow{2}{*}{\obj($\M$)} & \multirow{2}{*}{\mor($\M$)} & open maps & \multirow{4}{*}{%
                                                                                                 \hspace{-1mm}
  \begin{tikzpicture}[scale=0.8]	
	\node (Z) at (0,1) {$Z$};
	\node (X) at (-0.7,-0.8) {$T$};
	\node (Y) at (0.7,-0.8) {$T'$};
	\path[->,every node/.append style={align=center,anchor=south,font=\scriptsize}]
		(Z) edge node[sloped,above] {func.\\bisim\rlap{.}}  (X)
		(Z) edge node[sloped,above] {func.\\(bi)sim\rlap{.}}  (Y);	
st
\end{tikzpicture}\hspace{-2mm}}
   & \multirow{4}{*}{\rotatebox{270}{\hspace{-3mm}$\xrightarrow{\text{Lasota's}}$}}\\ 
   & maps & Def.~\ref{def:openmapsituation} &  &  & Def.~\ref{def:openmap} & \\
  \cmidrule{2-6}
  &\multirow{2}{*}{coalgebra} & $\map{G}{\C}{\C}$ & pointed G-coalg. & lax hom. & coalg. hom. & \\ 
   & & Def.~\ref{def:homsetOrder} & Sec.~\ref{subsec:coalgebras} & Def.~\ref{def:laxhomomorphism} & Def.~\ref{def:homomorphism} & \\
  \cmidrule[\heavyrulewidth]{2-7}
  \end{tabular}
  \hspace{-15mm} 
  \vspace{2mm}
  \caption{Two approaches to categorical (bi)simulations}
  \label{fig:categoricalApproaches}
  \vspace{-10mm}
\end{table}

We are naturally interested in the relationship between those two categorical
approaches to transition systems and bisimulations. A reduction of open maps
situations to coalgebra was given by Lasota using multi-sorted transition
systems~\cite{lasota02}. In this paper, we give the reduction in the other
direction: from the category $\LCoalg(TF)$ of pointed $TF$-coalgebras and lax
homomorphisms, we construct the path-category
$\Path$ and a functor $J:\Path\to\LCoalg(TF)$ such that
$\Path$-open morphisms coincide with strict 
homomorphisms, hence functional bisimulations. Here, $T$ is a functor describing
the branching behaviour and $F$ describes the input type, i.e.~the type of data
that is processed (e.g.~words or trees). This development is carried out
with the case where $T$ is a powerset-like functor, and covers transition
systems allowing non-deterministic branching.


The key concept in the construction of $\Path$ are $F$-\emph{precise maps}.
Roughly speaking in set, a map $f\colon X\to FY$ is $F$-precise if every $y\in
Y$ is used precisely once in $f$, i.e.~there is a unique $x$ such that $y$
appears in $f(x)$ and additionally $y$ appears precisely once in $f(x)$. Such an
$F$-precise map represents one deterministic step (of shape $F$). Then a path $P\in
\Path$ is a finite sequence of deterministic steps, i.e.~finitely many precise
maps. $J$ converts such a data into a pointed $TF$-coalgebra.
There are many existing notions of paths and traces in
coalgebra~\cite{BeoharK17,hasuo07,DBLP:conf/calco/JacobsS09,MiliusEA15}, which lack the notion of \emph{precise} map, which
is crucial for the present work.


Once we set up the situation $J\colon \Path\to\LCoalg(TF)$, we are on the
framework of open map bisimulations. Our construction of $\Path$ using
precise maps is justified by the characterisation theorem:
$\Path$-open morphisms and strict coalgebra homomorphisms coincide
(\autoref{thm:homIsOpen} and \autoref{thm:openIsHom}). This coincidence relies on the 
concept of path-reachable coalgebras, namely, coalgebras such that every state can be 
reached by a path. Under mild conditions, path-reachability is equivalent to an existing
notion in coalgebra, defined as the non-existence of a proper sub-coalgebra
(Section \ref{subsec:reachability}). Additionally, this characterization produces a canonical 
trace semantics for free, given in terms of paths (Section \ref{subsec:trace}).

We illustrate our reduction with several concrete situations: different classes
of non-deterministic top-down tree automata using analytic functors (Section
\ref{sec:treeAutomata}), Regular Nondeterministic Nominal Automata (RNNA), an
expressive automata notion living in nominal sets (Section \ref{sec:RNNA}),
multisorted transition systems, used in Lasota's work to construct a coalgebra
situation from an open map situation (Section \ref{subsec:Lasota}).

\paragraph*{Notation.} 
We assume basic categorical knowledge and notation (see e.g.~\cite{adamek04,awodey10}). The
cotupling of morphisms $f\colon A\rightarrow C$, $g\colon B\rightarrow C$ is denoted by
$[f,g]\colon A+B\rightarrow C$, and the unique morphsim to the terminal object is
$!\colon X\rightarrow 1$ for every $X$.

\section{Two categorical approaches for bisimulations}
\label{sec:prelim}
We introduce the two formalisms involved in the present paper: the open maps (Section 
\ref{subsec:openmaps}) and the coalgebras (Section \ref{subsec:coalgebras}). Those 
formalisms will be illustrated on the classic example of Labelled Transition Systems (LTSs).



\begin{definition}
Fix a set $A$, called the alphabet. A \emph{labelled transition system} is a 
triple $(S,i,\Delta)$ with $S$ a set of \emph{states}, $i \in S$ the \emph{initial 
state}, and $\Delta \subseteq S\times A\times S$ the \emph{transition 
relation}.  When $\Delta$
is obvious from the context,
we write $\trans s a {s'}$ to mean $(s,a,s')\in\Delta$.
\end{definition}

For instance, the tuple 
$(\{0,\cdots,n\},0,\{(k-1,a_k,k)~|~1\le k\le n\})$ is an LTS,
and called the {\em linear system} over the word $a_1\cdots a_n\in A^\star$.
To relate LTSs, one considers functions that preserves the structure
of LTSs:
\begin{definition}
\label{def:LTSmorphism}
A \emph{morphism of LTSs} from $(S,i,\Delta)$ to $(S',i', \Delta')$ is
a function $\map{f}{S}{S'}$ such that $f(i) = i'$ and for every
$(s,a,s')\in\Delta$, $(f(s),a,f(s'))\in\Delta'$. LTSs and morphisms of LTSs form
a category, which we denote by $\trsy$.
\end{definition}

Some authors choose other notions of morphisms (e.g.~\cite{joyal96}), allowing them to 
operate between LTSs with different alphabets for example. The usual way of comparing 
LTSs is by using simulations and bisimulations \cite{park81}. The former describes what it 
means for a system to have at least the behaviours of another, the latter describes that two 
systems have exactly the same behaviours. Concretely:

\begin{definition}
\label{def:LTSsimulation}
A \emph{simulation} from $(S,i,\Delta)$ to $(S',i',\Delta')$ is a relation
$R \subseteq S\times S'$ such that (1) $(i,i') \in R$, and (2) for every $\trans s a t$ and $(s,s') \in R$,
there is $t'\in S'$ such that $\trans{s'}a{t'}$ and $(t,t') \in R$.
Such a relation $R$ is a \emph{bisimulation} if
$R^{-1} = \{(s',s) \mid (s,s') \in R\}$ is also a simulation.
\end{definition}

Morphisms of LTSs are functional simulations, i.e.~functions between states whose 
graph is a simulation. So how to model (1) systems, (2) functional simulations and 
(3) functional bisimulations categorically? In the next two sections, we will 
describe known answers to this question, with open maps and coalgebra.
In both cases, it is possible to capture similarity and bisimilarity of two LTSs 
$T$ and $T'$.
Generally, a simulation is a (jointly monic) span of a functional
bisimulation and a functional simulation, and a bisimulation is a simulation
whose converse is also a simulation, as depicted in \autoref{fig:categoricalApproaches}.
%
%
%
%
%
Consequently, to understand
similarity and bisimilarity on a general level, it is enough to understand
functional simulations and bisimulations.

\vspace{-2mm}
\subsection{Open maps}
\label{subsec:openmaps}

The categorical framework of open maps \cite{joyal96} assumes functional
simulations to be already modeled as a category $\M$. For example, for $\M :=
\trsy$, objects are LTSs, and morphisms are functional simulations.
Furthermore, the open maps framework assumes another category $\bbP$ of `paths'
or `linear systems', together with a functor $J$ that tells how a `path' is to be
understood as a system:%
\begin{notheorembrackets}%
\begin{definition}[{\cite{joyal96}}]
\label{def:openmapsituation}
An \emph{open map situation} is given by categories $\M$ (`systems' with
`functional simulations') and $\bbP$
(`paths') together with a functor $J\colon \bbP \rightarrow \M $.
\end{definition}
\end{notheorembrackets}
For example with $\M:= \trsy$, we pick $\bbP := (A^\star,\le)$ to
be the poset of words over $A$ with prefix
order. Here, the functor $J$ maps a word $w\in A^\star$ to
the linear system over $w$, and
$w \le v$ to the evident  functional simulation $J(w\le v)\colon Jw\to Jv$.

In an open map situation $J\colon \bbP\to\M$, we can abstractly represent
the concept of a {\em run} in a system. A run of a path $w\in \bbP$ in a
system $T\in\M$ is simply defined to be an $\M$-morphism of type
$Jw\to T$. With this definition, each $\M$-morphism $h\colon T\to T'$
(i.e.~functional simulation) inherently transfers runs: given a run
$x\colon Jw\to T$, the morphism $h\cdot x\colon Jw\to T'$ is a run of
$w$ in $T'$.  In the example open map situation
$J\colon (A^\star,\le)\to\trsy$, a run of a path $w=a_1\cdots a_n\in A^\star$ in an
LTS $T=(S,i,\Delta)$ is nothing but a sequence of states
$x_0,\ldots,x_n \in S$ such that $x_0=i$ and
$x_{k-1}\xrightarrow{a_k} x_k$ holds for all $1\le k \le n$.


We introduce the concept of open map \cite{joyal96}. This is an
abstraction of the property posessed by {\em functional
  bisimulations}.  For LTSs $T=(S,i,\Delta)$ and $T'=(S',i',\Delta')$,
an $\trsy$-morphism $h\colon T\to T'$ is a functional bisimulation if
the graph of $h$ is a bisimulation. 
This implies the
following relationship between runs in $T$ and runs in $T'$. Suppose
that $w\le w'$ holds in $A^\star$, and a run $x$ of $w$ in $T$ is
given as in \eqref{eq:sysT}; here $n,m$ are lengths of $w,w'$
respectively.  Then for any run $y'$ of $w'$ in $T'$ extending
$h\cdot x$ as in \eqref{eq:sysTT}, there is a run $x'$ of $w'$
extending $x$, and moreover its image by $h$ coincides with $y'$ (that
is, $h\cdot x'=y'$).  Such $x'$ is obtained by repetitively applying
the condition of functional bisimulation.
\begin{align}
  \label{eq:sysT}
  \rightarrow\underbrace{\overbrace{ i\xrightarrow{w_1} x_1\xrightarrow{w_2} \cdots\xrightarrow{w_n} x_n}^x\xrightarrow{w'_{n+1}} x'_{n+1}\xrightarrow{w'_{n+2}}\cdots\xrightarrow{w'_{m}} x'_m}_{x'}
  \quad (\text{in $T$})
\\
  \label{eq:sysTT}
  \rightarrow\underbrace{ i'\xrightarrow{w_1} h(x_1)\xrightarrow{w_{2}} \cdots\xrightarrow{w_n} h(x_n)\xrightarrow{w'_{n+1}} y'_{n+1}\xrightarrow{w_{n+2}}\cdots\xrightarrow{w'_{m}} y'_m}_{y'}
  \quad (\text{in $T'$})
\end{align}
Observe that $y'$ extending $h\cdot x$ can be represented as
$y'\cdot J(w\le w')=h\cdot x$, and $x'$ extending $x$ as
$x'\cdot J(w\le w')=x$. From these, we conclude that if an
$\trsy$-morphism $h\colon T\to T'$ is a functional bisimulation,
then for any $w\le w'$ in $A^\star$ and run $x\colon Jw\to T$ and
$y'\colon Jw'\to T'$ such that $y'\cdot J(w\le w')=h\cdot x$, there is a
run $x'\colon Jw'\to T$ such that $x'\cdot J(w\le w')=x$ and $h\cdot x'=y'$
(the converse also holds if all states of $T$ are reachable). This
necessary condition of functional bisimulation can be rephrased in
any open map situation, leading us to the definition of open map.

\smallskip
\noindent
\begin{minipage}{.8\textwidth}
\begin{notheorembrackets}%
\begin{definition}[{\cite{joyal96}}]%
  \label{def:openmap}
  Let $J\colon \bbP\to \M$ be an open map situation.
  An $\M$-morphism $\map{h}{T}{T'}$ is said to be \emph{open} if for every morphism
  $\Phi\colon w\to w'\in\bbP$ making the square on the right commute, there is $x'$
  making the two triangles commute.
\end{definition}%
\end{notheorembrackets}
\end{minipage}%
\hfill%
\begin{minipage}{.18\textwidth}
  \hfill%
  \begin{tikzcd}
    Jw
    \arrow{r}{x}
    \arrow{d}[swap]{J\Phi}
    & T
    \arrow{d}{h}
    \\
    Jw'
    \arrow{r}{y'}
    \arrow[dashed]{ur}[sloped,above]{\exists x'}
    & T'
  \end{tikzcd}%
  \hspace{-2mm}%
\end{minipage}

\smallskip
Open maps are closed under composition and stable under pullback~\cite{joyal96}.

\subsection{Coalgebras}
\label{subsec:coalgebras}

The theory of G-coalgebras is another categorical framework to study
bimulations. The type of systems is modelled using an endofunctor
$\map{G}{\C}{\C}$ and a system is then a coalgebra for this functor, that is, a
pair of an object $S$ of $\C$ (modeling the state space), and of a morphism of
type $S \longrightarrow GS$ (modeling the transitions). For example for LTSs,
the transition relation is of type $\Delta \subseteq S\times A\times S$.
Equivalently, this can be defined as a function $\map{\Delta}{S}{\pow(A\times
  S)}$, where $\pow$ is the powerset. In other words, the transition relation is a
coalgebra for the Set-functor $\pow(A\times\arg)$. Intuitively, this coalgebra
gives the one-step behaviour of an LTS: $S$ describes the state space of the
system, $\pow$ describes the `branching type' as being non-deterministic, $
A\times S$ describe the `computation type' as being linear, and the function
itself lists all possible futures after one-step of computation of the system.
Now, changing the underlying category or the endofunctor allows to model
different types of systems. This is the usual framework of coalgebra, as
described for example in \cite{rutten00}.

Initial states are modelled coalgebraically by a pointing to the
carrier $i\colon I\to S$ for a fixed object $I$ in $\C$, describing the `type of
initial states' (see e.g.~\cite[Sec.~3B]{adamek13}). For example, an initial state of an LTS is the same as a
function from the singleton set $I:=\{\ast\}$ to the state space $S$. This
object $I$ will often be the final object of $\C$, but we will see other
examples later. In total, an \emph{$I$-pointed $G$-coalgebra} is a $\C$-object
$S$ together with morphisms $\alpha\colon S\to GS$ and $i\colon I\to S$. E.g.~an
LTS is an $I$-pointed $G$-coalgebra for $I=\{*\}$ and $GX = \pow(A\times X)$.

In coalgebra, functional bisimulations are the first class citizens to be
modelled as homomorphisms. The intuition is that those
preserve the initial state, and preserve and reflect the one-step
relation.

\smallskip
\noindent\begin{minipage}[b]{.75\textwidth}
\begin{definition}
\label{def:homomorphism}
An \emph{$I$-pointed $G$-coalgebra homomorphism} from
$I \xrightarrow{\,i\,} S \xrightarrow{\,\alpha\,} GS$ to
$I \xrightarrow{\,i'\,} S' \xrightarrow{\,\alpha'\,} GS'$ is a morphism
$\map{f}{S}{S'}$ making the right-hand diagram commute.
\end{definition}
\end{minipage}%
\hfill%
\begin{tikzcd}[sep=5mm, row sep=4mm, baseline=(Sp.base)]
  I
  \arrow{r}{i}
  \arrow{dr}[swap]{i'}
  & S
  \arrow{r}{\alpha}
  \arrow{d}{f}
  & GS
  \arrow{d}{Gf}
  \\
  & |[alias=Sp]| S'
  \arrow{r}{\alpha'}
  & GS'
\end{tikzcd}
\hspace*{-2.5mm}

\smallskip For instance, when $G=\pow(A\times \arg)$, one can easily see that a
function $f$ is a $G$-coalgebra homomorphism iff it is a functional bisimulation.
Thus, if we want to capture
functional simulations in LTSs, we need to weaken the condition of homomorphism to 
the inequality $Gf(\alpha (s))\subseteq \alpha'(f(s))$ (instead of equality). To
express this condition for general $G$-coalgebras, we introduce a partial order
$\sqsubseteq_{X,Y}$ on each homset $\C(X,GY)$ in a functorial
manner. 

\begin{definition}
  \label{def:homsetOrder}
  A \emph{partial order on $G$-homsets} is a functor
  ${\sqsubseteq}\colon \C^{\oper{op}}\times\C\to\Pos$ such that
  $U\cdot {\sqsubseteq}=\C(\arg,G\arg)$; here, $U\colon \Pos\to\Set$ is the
  forgetful functor from the category $\Pos$ of posets and monotone functions.
\end{definition}
\noindent The functoriality of $\sqsubseteq$
amounts to that $f_1\sqsubseteq f_2$ implies
$Gh\cdot f_1\cdot g\sqsubseteq Gh\cdot f_2\cdot g$.
\vspace{3mm}

\noindent\begin{minipage}[C]{.75\textwidth}
\begin{definition}
  \label{def:laxhomomorphism}
  Given a partial order on $G$-homsets,
  an \emph{$I$-pointed lax $G$-coalgebra homomorphism}
  $\map{f}{(S,\alpha,i)}{(S',\alpha',i')}$ is a morphism $f\colon S\to S'$
  making the right-hand diagram commute. The
  $I$-pointed $G$-coalgebras and lax \text{homomorphisms} form a category, denoted by
  $\LCoalg(I,G)$.
\end{definition}
\end{minipage}%
\hfill%
\begin{tikzcd}[sep=5mm, row sep=4mm]
  I
  \arrow{r}{i}
  \arrow{dr}[swap]{i'}
  & S
  \arrow{r}{\alpha}
  \arrow{d}{f}
  \descto[xshift=1mm]{dr}{\downinclusion}
  & GS
  \arrow{d}{Gf}
  \\
  & |[alias=Sp]| S'
  \arrow{r}[swap]{\alpha'}
  & GS'
\end{tikzcd}
\hspace*{-2.5mm}

\begin{conclusion}
  \label{ex:ltsOpen}
  In $\Set$, with $I=\{\ast\}$, $G = \pow(A\times\arg)$, define the order 
$f \sqsubseteq g$ in $\Set(X, \pow(A\times Y))$ 
iff for every $x \in X$, $f(x) \subseteq g(x)$. Then
  $\LCoalg(\{\ast\},\pow(A\times \arg)) = \trsy$.
  In particular, we have an open map situation \\[1mm]
  \hspace*{8mm}\(
   ~\bbP = (A^\star,\le) ~~\overset{J}{\to}~~ \M =\trsy = \LCoalg(\{\ast\}, \pow(A\times \arg))
  \)\\[1mm]
and the open maps are precisely the coalgebra homomorphisms (for reachable LTSs).
In this paper, we will construct a path category $\bbP$ for more general
$I$ and $G$, such that the open morphisms are precisely the coalgebra homomorphisms.
\end{conclusion}

\section{The open map situation in coalgebras}
\label{sec:pathcat}
Lasota's construction~\cite{lasota02} transforms an open map situation $J\colon
\bbP \to \M$ into a functor $G$ (with a partial order on $G$-homsets), together
with a functor $\oper{Beh}\colon \M\to \LCoalg(I,G)$ that sends open maps to
$G$-coalgebra homomorphisms (see section~\ref{subsec:Lasota} for details). In
this paper, we provide a construction in the converse direction for functors $G$
of a certain shape.

As exemplified by LTSs, it is a common pattern that $G$ is the composition
$G=TF$ of two functors~\cite{hasuo07}, where $T$ is the branching type (e.g.~partial,
or non-deterministic) and $F$ is the data type, or the `linear behaviour'
(words, trees, words modulo $\alpha$-equivalence). If we
instantiate our path-construction to $T=\pow$ and $F=A\times \arg$, we obtain the
known open map situation for LTSs (\autoref{ex:ltsOpen}).

Fix a category $\C$ with pullbacks, functors $T,F:\C\to\C$, an object
$I\in\C$ and a partial order $\sqsubseteq^T$ on $T$-homsets. They
determine a coalgebra situation $(\C,I,TF,\sqsubseteq)$ where
$\sqsubseteq$ is the partial order on $TF$-homsets defined by
${\sqsubseteq_{X,Y}}={\sqsubseteq_{X,FY}^T}$.
Under some conditions on $T$ and $F$, we construct a path-category
$\Path(I,F+1)$ and an open map situation
$\Path(I,F+1)\hookrightarrow \LCoalg(I,TF)$ where
$TF$-coalgebra homomorphisms and
$\Path(I,F+1)$-open morphisms coincide. 

\subsection{Precise morphisms}
\label{subsec:preciseMorph}
While the path category is intuitively clear for $FX=A\times X$, it is not for
inner functors $F$ that model tree languages. For example for $FX=A + X\times X$, a
$\pow F$-coalgebra models transition systems over binary trees with leaves
labelled in $A$, instead of over words. Hence, the paths should be these kind of
binary trees. We capture the notion of tree like shape (``every node in a tree
has precisely one route to the root'') by the following abstract definition:

\begin{definition}
  For a functor $F\colon \C\to \C$, a morphism $s\colon S\to FR$ is called
  \emph{$F$-precise} if for all $f,g,h$ the following implication holds:
  \[
    \begin{tikzcd}[sep=4mm]
      S
      \arrow{r}{f}
      \arrow{d}[swap]{s}
      & FC
      \arrow{d}{Fh}
      \\
      FR
      \arrow{r}{Fg}
      & FD
    \end{tikzcd}
    \quad\overset{\exists d\,}{\Longrightarrow}\quad
    \begin{tikzcd}[sep=4mm]
      S
      \arrow{r}{f}
      \arrow{d}[swap]{s}
      & FC
      \\
      FR
      \arrow{ur}[swap]{Fd}
    \end{tikzcd}
    \&
    \begin{tikzcd}[sep=4mm]
      & C
      \arrow{d}{h}
      \\
      R
      \arrow{r}{g}
      \arrow{ur}{d}
      & D
    \end{tikzcd}
  \]
\end{definition}
\begin{remark}
\label{rem:PreciseWeakPullbacks}
  If $F$ preserves weak pullbacks, then a morphism $s$ is $F$-precise iff it
  fulfils the above definition for $g=\id$.
\end{remark}
\begin{example}
  \label{exPreciseMap}
  Intuitively speaking, for a polynomial $\Set$-functor $F$, a map $s\colon 
  S\rightarrow
  FR$ is $F$-precise iff every element of $R$ is mentioned precisely once in the
  definition of the map $f$. For example, for $FX = A\times X + \{\bot\}$, the
  case needed later for LTSs, a map $f\colon X\to FY$ is precise iff for every
  $y \in Y$, there is a unique pair $(x,a) \in X\times A$ such that $f(x) =
  (a,y)$. For $FX=X\times X +\{\bot\}$ on $\Set$, the map $f\colon X\to FY$ in
  \autoref{figExPrecise} is not $F$-precise, because $y_2$ is used three times
  (once in $f(x_2)$ and twice in $f(x_3)$), and $y_3$ and $y_4$ do not occur in
  $f$ at all. However, $f'\colon X\to FY'$ is $F$-precise because every element
  of $Y'$ is used precisely once in $f'$, and we have that $Fh\cdot f' = f$.
  Also note that $f'$ defines a forest where $X$ is the set of roots, which is
  closely connected to the intuition that, in the $F$-precise map $f'$, from
  every element of $Y'$, there is precisely one edge up to a root in $X$.
\end{example}

\begin{figure}[t] \centering
  \begin{tikzpicture}[xscale = 1,yscale=0.6]
    \foreach \prefix/\varname/\labeltext/\count/\xshift in
    {x/x/X/4/0,y/y/Y/4/2,xcopy/x/X/4/5,
      yprime/{\ensuremath{y'}}/{\ensuremath{Y'}}/4/7,ycopy/y/Y/4/9} {
      \foreach \n in {1,...,\count} {
        \node[named vertex] (\prefix\n) at (\xshift cm,-\n) {$\varname_\n$};
      }
      \node[set,fit=(\prefix1) (\prefix\count),label={[name=\prefix]$\labeltext$}] (\prefix domain) {};
    }
    \path[myedge] (x2) to (y1) ;
    \path[myedge] (x2) to[bend left=20] (y2.north west) ;
    \path[myedge] (x3) to[bend left] (y2) ;
    \path[myedge] (x3) to[bend right] (y2) ;
    \foreach \node in {x1,x4,xcopy1,xcopy4} {
      \path[myedge] (\node) to node[pos=1]{$\bot$} +(8mm,0) ;
    }
    \path[myedge] (xcopy2) to (yprime1) ;
    \path[myedge] (xcopy2) to (yprime2) ;
    \path[myedge] (xcopy3) to (yprime3) ;
    \path[myedge] (xcopy3) to (yprime4) ;
    \path[myedge,|->] (yprime1) to (ycopy1);
    \path[myedge,|->] (yprime2) to[bend left=15] (ycopy2);
    \path[myedge,|->] (yprime3) to[bend right=0] (ycopy2);
    \path[myedge,|->] (yprime4) to[bend right=18] (ycopy2);
    \path[every label]
    (x) edge[draw=none] node {$f$} (y)
    (xcopy) edge[draw=none] node {$f'$} (yprime)
    (yprime) edge[draw=none] node {$h$} (ycopy)
    ;
  \end{tikzpicture}
  \caption{A non-precise map $f$ that factors through the $F$-precise $f'\colon
    X\to Y'\times Y' +\{\bot\}$}
  \label{figExPrecise}
  \vspace{-5mm}
\end{figure}

So when transforming a non-precise map into a precise map, one duplicates
elements that are used multiple times and drops elements that are not used. We
will cover functors $F$ for which this factorization pattern provides
$F$-precise maps. If $F$ involves unordered structure, this factorization needs
to make choices, and so we restrict the factorization to a class
$\S$ of objects that have that choice-principle (see \autoref{ex:whyClassS} later):\\

\noindent
\begin{minipage}{.78\textwidth}%
\begin{definition} \label{defPreciseFactor}
  Fix a class of objects $\S\subseteq \obj\C$ closed under isomorphism. We say
  that \emph{$F$ admits precise factorizations w.r.t.~$\S$}  if for every
  $f\colon S\rightarrow FY$ with $S\in \S$, there exist $Y'\in \S$, $h\colon Y'\rightarrow Y$ and
  $f'\colon S\rightarrow FY'$ $F$-precise with $Fh\cdot f' = f$.
\end{definition}
\end{minipage}\hfill%
\begin{minipage}{.210\textwidth}%
\vspace{-7mm}
\hfill%
\[
    \begin{tikzcd}
      S
      \arrow{dr}[swap]{\forall f}
      \arrow[dashed]{r}{\exists f'}
      & FY'
      \arrow{d}{Fh}
      \\
      &  FY
    \end{tikzcd}
  \]
  \end{minipage}\\
  
\noindent For $\C=\Set$, $\S$ contains all sets. However for
the category of nominal sets, $\S$ will only contain the strong nominal sets
(see details in \autoref{sec:RNNA}).

\begin{remark} \label{remPreciseFactor} Precise morphisms are essentially
  unique. If $f_1\colon X\to FY_1$ and $f_2\colon X\to FY_2$ are $F$-precise and
  if there is some $h\colon Y_1\to Y_2$ with $Fh\cdot f_1 = f_2$, then $h$ is an
  isomorphism. Consequently, if $f\colon S\to FY$ with $S\in \S$ is $F$-precise and
  $F$-admits precise factorizations, then $Y\in \S$.
\end{remark}

\noindent Functors admitting precise factorizations are closed under basic
constructions:
\begin{proposition} \label{propFactorizationClosure}
  The following functors admit precise factorizations w.r.t.~$\S$:\\
 1.~ Constant functors, if $\C$ has an initial object $0$ and $0 \in \S$.\\
 2.~ $F\cdot F'$ if $F\colon \C\to \C$ and  $F'\colon \C\to \C$ do so.\\
 3.~ $\prod\limits_{i\in I} F_i$, if all $(F_i)_{i\in I}$ do so and
      $\S$ is closed under $I$-coproducts.\\
 4.~ $\coprod\limits_{i\in I} F_i$, if all $(F_i)_{i\in I}$ do so, 
      $\C$ is $I$-extensive and $\S$ is closed under $I$-coproducts.\\
 5.~ Right-adjoint functors, if and only if its left-adjoint preserves $\S$-objects.
\end{proposition}
\begin{example} \label{powersetNotPrecise} When $\C$ is infinitary extensive 
and $\S$ is closed under coproducts, every polynomial endofunctor 
$\map{F}{\C}{\C}$ admits precise factorizations w.r.t. $\S$. This is in particular 
the case for $\C = \S = \Set$. In this case, we shall see later (Section 
\ref{sec:treeAutomata}) that many other \Set-functors, e.g. the bag functor $\B$, 
where $\B(X)$ is the set of finite multisets, have precise factorizations. 
  In contrast, $F=\pow$ does not
  admit precise factorizations, and if $f\colon X\to \pow Y$ is $\pow$-precise,
  then $f(x) = \emptyset$ for all $x\in X$.
\end{example}

\subsection{Path categories in pointed coalgebras} \label{secPathCategory} We
define a path for $I$-pointed $TF$-coalgebras as a tree according to $F$.
Following the observation in \autoref{exPreciseMap}, one layer of the tree is
modelled by a $F$-precise morphism and hence a path in a $TF$-coalgebra is
defined to be a finite sequence of $(F+1)$-precise maps, where the $\arg+1$ comes
from the dead states w.r.t.~$T$; the argument is given later in
\autoref{whyPlus1} when reachability is discussed. Since the $\arg+1$ is not
relevant yet, we define $\Path(I,F)$ in the following and will use
$\Path(I,F+1)$ later.
\begin{figure}[t]
  \centering
  \begin{tikzpicture}[xscale = 1,yscale =0.5]
    \begin{scope}[xshift=-2cm,yshift=5mm]
      \node[] (si1) at (0,0) {$\ast$};
      \node[set,fit=(si1),label={[name=S0]$P_0$}] {};
    \end{scope}
    \begin{scope}[xshift=0]
      \node[vertex] (s01) at (0,0) {};
      \node[vertex] (s02) at (0,1) {};
      \node[set,fit=(s01) (s02),label={[name=S1]$P_1$}] {};
    \end{scope}
    \begin{scope}[xshift=2cm,yshift=-5mm]
      \node[vertex] (s11) at (0,0) {};
      \node[vertex] (s12) at (0,1) {};
      \node[vertex] (s13) at (0,2) {};
      \node[set,fit=(s11) (s13),label={[name=S2]$P_2$}] {};
    \end{scope}
    \begin{scope}[xshift=4cm,yshift=-5mm]
      \node[vertex] (s21) at (0,0) {};
      \node[vertex] (s22) at (0,1) {};
      \node[vertex] (s23) at (0,2) {};
      \node[set,fit=(s21) (s23),label={[name=S3]$P_3$}] {};
    \end{scope}
    \begin{scope}[xshift=6cm,yshift=0mm]
      \node[vertex] (s41) at (0,0) {};
      \node[vertex] (s42) at (0,1) {};
      \node[set,fit=(s41) (s42),label={[name=S4]$P_4$}] {};
    \end{scope}
    \path[myedge] (si1) to (s01) ;
    \path[myedge] (si1) to (s02) ;
    \path[myedge] (s01) to (s11) ;
    \path[myedge] (s01) to (s12) ;
    \path[myedge] (s02) to node {$a$} (s13) ;
    \path[myedge] (s11) to node {$a$} (s21) ;
    \path[myedge] (s12) to (s22) ;
    \path[myedge] (s12) to (s23) ;
    \path[myedge] (s13) to node[pos=1]{$\bot$} +(8mm,0) ;
    \path[myedge] (s23) to node[pos=1]{$\bot$} +(8mm,0) ;
    \path[myedge] (s21) to node[pos=1]{$\bot$} +(8mm,0) ;
    \path[myedge] (s22) to (s41) ;
    \path[myedge] (s22) to (s42) ;
    \path[every label]
      (S0) edge[draw=none] node {$p_0$} (S1)
      (S1) edge[draw=none] node {$p_1$} (S2)
      (S2) edge[draw=none] node {$p_2$} (S3)
      (S3) edge[draw=none] node {$p_3$} (S4)
      ;
  \end{tikzpicture}
  \caption{
    A path of length $4$ for 
  $FX = \{a\}\times X + X\times X +\{\bot\}$ with $I=\{ \ast \}$.
  }
  \label{figPathExample}
  \vspace{-5mm}
\end{figure}
For simplicity, we write $\vec{X}_n$ for finite families $(X_k)_{0\le k < n}$.
\begin{definition}
  The category $\Path(I,F)$ consists of the following.
  An object is $(\vec P_{n+1},\vec p_n)$ for an $n\in \N$ with
    $P_0 = I$ and $\vec p_n$ a family of $F$-precise maps $(p_k\colon P_k\to
    FP_{k+1})_{k<n}$. We say that $(\vec P_{n+1},\vec p_n)$ is a \emph{path
      of length} $n$.
  A morphism
    $\vec \phi_{n+1}\colon (\vec P_{n+1},\vec p_{n}) \to (\vec Q_{m+1},\vec
    q_{m})$, $m\ge n$, is a family $(\phi_k\colon P_k \to Q_k)_{k\le n}$ with $\phi_0=\id_I$ and
    $q_k\cdot \phi_k = F\phi_{k+1}\cdot p_k$ for all $0 \le k \le n$.
\end{definition}

\begin{example}
Paths for $FX = A\times X + 1$ and $I = \{*\}$ singleton are as follows. First, a 
map $\map{f}{I}{FX}$ is precise iff (up-to isomorphism) either $X = I$ and $f(\ast) = (a,
\ast)$ for some $a \in A$; or $X = \varnothing$ and $f(\ast) = \bot$. Then a path is 
isomorphic to an object of the form: $P_i = I$ for $i \leq k$, $P_i = \varnothing$ for $i > k$, 
$p_i(\ast) = (a_i,\ast)$ for $i < k$, and $p_k(\ast) = \bot$. A path is the same as a 
word, plus some ``junk'', concretely, a word in $A^\star.\bot^\star$. For LTSs,
an object in $\Path(I,F)$ with $FX=A\times X$ is simply a word in $A^\star$.
For a more complicated functor,
\autoref{figPathExample} depicts a path of length 4, which is a tree for the
signature with one unary, one binary
symbol, and a constant. The layers of the
tree are the sets $\vec{P}_4$. Also note that since every $p_i$ is $F$-precise,
there is precisely one route to go from every element of a $P_k$ to~$\ast$.
\end{example}

\begin{remark} \label{remarkPathCategory} The inductive continuation of
  \autoref{remPreciseFactor} is as follows. Given a morphism $\vec{\phi}_{n+1}$
  in $\Path(I,F)$, since $\phi_0$ is an isomorphism, then $\phi_k$ is an isomorphism
  for all $0\le k \le n$. If $F$ admits precise factorizations and if $I\in \S$,
  then for every path $(\vec
  P_{n+1},\vec p_{n})$, all $P_k$, $0\le k\le n$, are in $\S$.
\end{remark}
\begin{remark} \label{uniqueFactorization} If in \autoref{defPreciseFactor}, the
  connecting morphism $h\colon Y'\to Y$ uniquely exists, then it follows by
  induction that the hom-sets of
  $\Path(I,F)$ are at most singleton. This is the case for all polynomial
  functors, but not the case for the bag functor on sets (discussed
  in \autoref{sec:treeAutomata}).
\end{remark}

\vspace{-2mm}
\noindent
\begin{minipage}{.78\textwidth}%
\begin{definition}
  The \emph{path poset} $\PathOrd(I,F)$ is the set $ \coprod_{0 \le
    n}\C(I,F^n1)$ equipped with the order: for $u\colon I\to F^n1$ and
  $v\colon I\to F^m1$, we define $u\le v$ if $n \le m$ and 
  $F^n(!)\cdot v = u$.
\end{definition}%
\end{minipage}\hfill%
\begin{minipage}{.210\textwidth}%
\hfill%
\begin{tikzcd}[column sep=3mm,row sep=6mm]
      & \!F^{n}F^{m-n} 1
      \arrow{d}{F^n !}
      \\
      I
  \arrow[bend left=10]{ur}{v}
      \arrow{r}[yshift=1pt]{u}
      & F^n 1
\end{tikzcd}\hspace*{-2mm}%
\end{minipage}\\
So $u\le v$ if $u$ is the truncation of $v$ to $n$ levels. This matches the
morphisms in $\Path(I,F)$ that witnesses that one path is prefix of another:

\begin{proposition}
\label{prop:PathToPathOrd}
\noindent 1.~The functor $\Comp\colon \Path(I,F)\to \PathOrd(I,F)$ defined by
  \(
    I= P_0 \overset{p_0}{\rightarrow} FP_1\cdots\rightarrow
    F^{n}P_{n}\overset{F^{n}!}\rightarrow F^{n}1 \text{ on
      }(\vec P_{n+1},\vec p_{n})
  \)
  is full, and reflects isos.\\
\noindent 2.~If $F$ admits precise factorizations w.r.t.~$\S$ and $I\in \S$, then $\Comp$ is
  sujective.
\noindent 3.~If additionally $h$ in Def.~\ref{defPreciseFactor} is
    unique, then $\Comp$ has a right-inverse.
\end{proposition}


In particular, $\PathOrd(I,F)$ is $\Path(I,F)$ up to isomorphism. In the
instances, it is often easier to characterize $\PathOrd(I,F)$. This also shows
that $\Path(I,F)$ contains the elements -- understood as morphisms from $I$ --
of the finite start of the final chain of $F$: $
  1 \xleftarrow{!}
  F1 \xleftarrow{F!}
  F^21 \xleftarrow{F^2!}
  F^31 \xleftarrow{}\cdots.
$

\begin{example}
When $FX = A\times X + 1$, $F^n1$ is isomorphic to the set of words in $A^\star.\bot^\star$ of length $n$. 
Consequently, $\PathOrd(I,F)$ is the set of words in $A^\star.\bot^\star$, equipped with the prefix order. 
In this case, $\Comp$ is an equivalence of categories.
\end{example}

\subsection{Embedding paths into pointed coalgebras}
The paths $(\vec P_{n+1},\vec p_n)$ embed into $\LCoalg(I,TF)$ as one 
expects it for examples like \autoref{figPathExample}: one takes the disjoint
union of the $P_k$, one has the pointing $I=P_0$ and the linear structure of $F$
is embedded into the branching type $T$.

During the presentation of the results, we require $T$, $F$, and $I$ to have
certain properties, which will be introduced one after the other. The full list
of assumptions is summarized in \autoref{table:axioms}:\\
\indent \eqref{axFactor} -- The main theorem will show that coalgebra 
homomorphisms in $\LCoalg(I,TF)$ are the
open maps for the path category $\Path(I, F+1)$. So from now on, we assume that
$\C$ has finite coproducts and to use the results from the previous sections, we fix
a class $\S\subseteq \obj \C$ such that $F+1$ admits precise
factorizations w.r.t.~$\S$ and that $I\in \S$.\\
\indent \eqref{axJointlyepic} -- Recall, that a family of morphisms 
$(e_i\colon X_i \to Y)_{i\in I}$ with common
codomain is called jointly epic if for $f,g\colon Y\to Z$ we have that $f\cdot
e_i = g\cdot e_i~\forall i\in I$ implies $f=g$. For \Set, this means, that every
element $y\in Y$ is in the image of some $e_i$. Since we work with
partial orders on $T$-homsets, we also need the generalization of this property if 
$f\sqsubseteq g$ are of the form $Y\to TZ'$.\\
\indent \eqref{axUnitBot} -- In this section, we encode paths as a pointed coalgebra by 
constructing a functor $J\colon \Path(I,F+1)\hookrightarrow \LCoalg(I,TF)$.
For that we need to embed the linear behaviour $FX+1$ into $TFX$. This is done
by a natural transformation $[\eta,\bot]\colon \Id+1\to T$, and we require that
$\bot\colon 1\to T$ is a bottom element for $\sqsubseteq$.

\begin{example}
For the case where $T$ is the powerset functor $\pow$, $\eta$ is given by the unit 
$\eta_X(x) = \{x\}$, and $\bot$ is given by empty sets $\bot_X(\ast) = \varnothing$.
\end{example}

\begin{definition} \label{def:inclusionFunctor}
  We have an inclusion functor $J\colon \Path(I,F+1)\hookrightarrow
  \LCoalg(I,TF)$ that maps a path $(\vec P_{n+1},\vec p_n)$ to an $I$-pointed 
  $TF$-coalgebra on $\coprod\vec P_{n+1} := \coprod_{0\le k\le n} P_k$. The
  pointing is given by $\inj_0\colon I =P_0\to \coprod\vec P_{n+1} $ and the
  structure by:
  \[
    \smash{\coprod_{\mathclap{0 \le k < n}}} P_k + P_{n}
    \xrightarrow{[(F\inj_{k+1}+1)\cdot p_k]_{0\le k< n}+ !}
    F\coprod\vec P_{n+1}+ 1
    \xrightarrow{[\eta,\bot]}
    TF\coprod\vec P_{n+1}.
  \]
\end{definition}

\begin{example}
In the case of LTSs, a path, or equivalently a word $a_1...a_k.\bot...\bot \in 
A^\star.\bot^\star$, is mapped to the finite linear system over
$a_1...a_k$ (see Section \ref{subsec:openmaps}), seen as a coalgebra 
(see Section \ref{subsec:coalgebras}).
\end{example}

\begin{proposition}\label{lem:stepwiseFromPath}
  Given a morphism $[x_k]_{k\le n}\colon \coprod\vec P_{n+1} \to X$ for
  some system $(X,\xi,x_0)$ and a path $(\vec P_{n+1},\vec p_{n})$, we have
  \vspace{-1mm}
  \[
    \begin{array}{c}
    J(\vec P_{n+1},\vec p_{n}) \xrightarrow{[x_k]_{k\le n}} (X,\xi,x_0)
      \\[1mm]
      \text{a run in }\LCoalg(I,TF)
    \end{array}
    ~\Longleftrightarrow
    \forall k < n\colon\hspace{-4mm}
    \begin{tikzcd}[column sep=5mm,row sep=3mm]
      P_k
      \arrow{rr}{x_k}
      \arrow{d}[swap]{p_k}
      & & X
      \arrow{d}{\xi}
      \\
      FP_{k+1}+1
      \arrow{r}[yshift=1.5mm]{Fx_{k+1}+1}
      \descto[xshift=-2mm]{urr}{$\sqsubseteq$}
      & FX+1
      \arrow{r}[yshift=1.5mm]{[\eta,\bot]_X}
      & TFX.
    \end{tikzcd}
  \]
\end{proposition}
\vspace{-2mm}
Also note that the pointing $x_0$ of the coalgebra is necessarily the first
component of any run in it. In a run $[x_k]_{k\le n}$, $p_k$
corresponds to an edge from $x_k$ to $x_{k+1}$.

\begin{example}
  For LTSs, since the $P_k$ are singletons, $x_k$ just picks the $k$th state of
  the run. The right-hand side of this lemma describes that this is a run iff
  there is a transition from the $k$th
  state and the $(k+1)-$th state.
\end{example}

\subsection{Open morphisms are exactly coalgebra homomorphisms}
In this section, we prove our main contribution, namely that $\Path(I,F+1)$-open maps in 
$\LCoalg(I,TF)$ are exactly coalgebra homomorphisms. For the first direction of the main 
theorem, that is, that coalgebra homomorphisms are open, we need two extra axioms:\\
\indent \eqref{axJoin} -- describing that the order on $\C(X,TY)$ is point-wise. This 
holds for the powerset because every set is the union of its singleton subsets.\\
\indent  \eqref{axChoice} -- describing that $\C(X,TY)$ admits a choice-principle. This holds 
for the powerset because whenever $y \in h[x]$ for a map $h\colon X\to Y$ and 
$x\subseteq X$, then there is some $\{x'\}\subseteq x$ with $h(x') = y$.

\begin{theorem}
  \label{thm:homIsOpen}
  Under the assumptions of \autoref{table:axioms}, a coalgebra homomorphism in
  $\LCoalg(I,TF)$ is $\Path(I,F+1)$-open.
\end{theorem}

\begin{table}[t]
  \centering
  \def\arraystretch{1.2}
  \begin{tabular}{@{}l@{\hspace{2mm}}l@{\hspace{4mm}}l@{}}
    \toprule
    $F$ & \axlabel{axFactor} & $F+1$ admits precise factorizations, w.r.t.~$\S$ and $I\in \S$  \\
    \midrule
    $T$ & \axlabel{axJointlyepic} & If $(e_i\colon X_i\to Y)_{i\in I}$ jointly epic, then $f\cdot e_i\sqsubseteq g\cdot e_i$ for all $i\in I$ $\Rightarrow$ $f\sqsubseteq g$.\\
        & \axlabel{axUnitBot} & $[\eta,\bot]\colon \Id+1\to T$, with $\bot_Y\cdot !_X \sqsubseteq f$ for all $f\colon X\to TY$
                  \\
        & \axlabel{axJoin} & For every $f\colon X\to TY$, $X\in \S$, \\
        &  & $f = \bigsqcup\{[\eta,\bot]_Y\cdot f'\sqsubseteq f \mid f'\colon X\to Y+1\}$\\
        & \axlabel{axChoice} &     $~~\forall A \in \S~~$
        \begin{tikzcd}[ampersand replacement=\&]
      A
      \arrow{r}{x}
      \arrow{d}[swap]{y}
      \descto[sloped,rotate=90]{dr}{$\sqsubseteq$}
      \& T X
      \arrow{d}{T h}
      \\
      Y+1
      \arrow{r}{[\eta,\bot]_Y}
      \& T Y
    \end{tikzcd}
    $~~\overset{\exists x'}\Longrightarrow~~$
    \tikzset{external/export=false}
    \begin{tikzcd}[ampersand replacement=\&,row sep=3mm]
      |[yshift=0mm]|
      A
      \arrow[bend left=10,overlay]{rr}{\smash{x}}
      \arrow[dashed]{r}{x'}
      \arrow[bend right=10]{dr}[swap]{y}
      \descto[sloped,rotate=90,xshift=-0.5mm]{rr}{$\sqsubseteq$}
      \& |[yshift=-4mm]| X+1
      \arrow{r}[swap,near start]{[\eta,\bot]_X}
      \arrow{d}{h+1}
      \& T X
      \arrow{d}{T h}
      \\
      \& Y+1
      \arrow{r}{[\eta,\bot]_Y}
      \& T Y
    \end{tikzcd}
    \\
    \bottomrule
  \end{tabular}\\[1mm]
  \caption{Main assumptions on $F,T\colon \C\to \C$, $\sqsubseteq^T$,
    $\S\subseteq \obj \C$}
  \label{table:axioms}
  \vspace{-7mm}
\end{table}

The converse is not true in general, because intuitively, open maps reflect
runs, and thus only reflect edges of reachable states, as we have seen in 
Section \ref{subsec:openmaps}. The notion of
a state being reached by a path is the following:
\begin{definition} \label{defPathReachable}
  A system $(X,\xi,x_0)$ is \emph{path-reachable} if the family of runs $[x_k]_{k\le
    n}\colon J(\vec P_{n+1},\vec p_n)\to (X,\xi,x_0)$ (of paths from
  $\Path(I,F+1)$) is jointly epic.
\end{definition}
\begin{example}
For LTSs, this means that every state in $X$ is reached by a run, that is,
there is a path from the initial state to every state of $X$. 
\end{example}
\begin{remark} \label{whyPlus1}
  In \autoref{defPathReachable}, it is crucial that we consider $\Path(I,F+1)$
  and not $\Path(I,F)$ for functors incorporating `arities $\ge 2$'. This does
  not affect the example of LTSs, but for $I=1$, $FX=X\times X$ and $T=\pow$ in
  \Set, the coalgebra $(X,\xi,x_0)$ on $X=\{x_0,y_1,y_2,z_1,z_2\}$ given by
  \(
    \xi(x_0) = \{ (y_1,y_2) \},
    ~~
    \xi(y_1) = \{ (z_1,z_2) \},
    ~~
    \xi(y_2) = \xi(z_1) = \xi(z_2) = \emptyset
  \)
  is path-reachable for $\Path(I,F+1)$. There is no run of a length 2 path from
  $\Path(I,F)$, because $y_2$ has no successors, and so there is no path to
  $z_1$ or to $z_2$.
\end{remark}
\begin{theorem} \label{thm:openIsHom}
  Under the assumptions of \autoref{table:axioms}, if
$(X,\xi,x_0)$ is path-reachable, then an open morphism $h\colon (X,\xi,x_0)
\to (Y,\zeta,y_0)$ is a coalgebra homomorphism.
\end{theorem}

\subsection{Connection to other notions of reachability}
\label{subsec:reachability}
There is another concise notion for reachability in the coalgebraic literature \cite{adamek13}.
\begin{definition}
  A \emph{subcoalgebra} of $(X,\xi,x_0)$ is a coalgebra homomorphism $h\colon
  (Y,\zeta,y_0) \to (X,\xi,x_0)$ that is carried by a monomorphism $h\colon
  X\rightarrowtail Y$. Furthermore $(X,\xi,x_0)$ is called \emph{reachable} if
  it has no proper subcoalgebra, i.e.~if any subcoalgebra $h$ is an isomorphism.
\end{definition}
Under the following assumptions, this notion coincides with the path-based
definition of reachability (\autoref{defPathReachable}).
\begin{assumption} \label{ass:reachability} For the present
  \autoref{subsec:reachability}, let $\C$ be cocomplete, have
  (epi,mono)-factorizations and wide pullbacks of monomorphisms.
\end{assumption}
The first direction follows directly from \autoref{thm:homIsOpen}:
\begin{proposition}
\label{prop:JointlyEpicToNoProper}
  Every path-reachable $(X,\xi,x_0)$ has no proper subcoalgebra.
\end{proposition}
For the other direction it is needed that $TF$ preserves arbitrary
intersections, that is, wide pullbacks of monomorphisms. In \Set, this means
that for a family $(X_i\subseteq Y)_{i\in I}$ of subsets we have
$\bigcap_{i\in I} TFX_i = TF\bigcap_{i\in I} X_i$ as subsets of $TFY$.

\begin{proposition}
  \label{prop:noProperReachable}
  If, furthermore, for every monomorphism $m\colon Y\to Z$, the function 
  $\C(-,Tm)\colon$ $\C(X,TY)\to \C(X,TZ)$ reflects joins and if $TF$ preserves
  arbitrary intersections, then a reachable 
  coalgebra $(X,\xi,x_0)$ is also path-reachable.
\end{proposition}
\noindent
All those technical assumptions are satisfied in the case of LTSs, and will also 
be satisfied in all our instances in \autoref{sec:instances}. 

\vspace{-2mm}
\subsection{Trace semantics for pointed coalgebras}
\label{subsec:trace}
The characterization from \autoref{thm:homIsOpen} and \autoref{thm:openIsHom} 
points out a natural way of defining a trace semantics for pointed coalgebras. 
Indeed, the paths category $\Path(I,F+1)$ provides a natural way of defining the 
runs of a system. A possible way to go from runs to trace semantics is to describe 
accepting runs as the subcategory $J'\colon \Path(I,F) \hookrightarrow 
\Path(I,F+1)$. We can define the \emph{trace semantics} of a system $(X,\xi,x_o)$ as the set:
\begin{align*} 
\trace(X,\xi,x_0) = \{\Comp(\vec P_{n+1},\vec p_n) \mid &\exists \text{ run } [x_k]_{k\le
    n}\colon JJ'(\vec P_{n+1},\vec p_n)\to (X,\xi,x_0)\\
    ~ & \text{ with } (\vec P_{n+1},\vec p_n) \in \Path(I,F)\}
\end{align*}
Since $\Path(I,F)$-open maps preserve and reflect runs, we have the following:

\begin{corollary}
  $\trace\colon \LCoalg(I,TF)\to (\pow(\PathOrd(I,F)), \subseteq)$ is a functor
and if $\map{f}{(X,\xi,x_0)}{(Y,\zeta,y_0)}$ is $\Path(I,F+1)$-open, then $\trace(X,\xi,x_0) = 
\trace(Y,\zeta,y_0)$.
\end{corollary}
Let us look at two LTS-related examples (we will describe some others in the next 
section). First, for $FX = A\times X$. The usual trace semantics is given by all the 
words in $A^\star$ that are labelled of a run of a system. This trace semantics
is obtained because $\PathOrd(I,F) = \coprod_{n\ge 0} A^n$ and because $\Comp$
maps every path to its underlying word. Another example is given for 
$FX = A\times X + \{\checkmark\}$, where $\checkmark$ marks final states. 
In this case, a path in $\Path(I,F)$ of length $n$ is either a path that can
still be extended or encodes less than $n$ steps to an accepting state $\checkmark$.
This obtains the trace semantics containing the
set of accepted words, as in automata theory, plus the set of possibly infinite runs.

\section{Instances} 
\label{sec:instances}
\subsection{Analytic functors and tree automata}
\label{sec:treeAutomata}
In Example \ref{powersetNotPrecise}, we have seen that every polynomial 
\Set-functors, in particular the functor $X \mapsto A\times X$, 
has precise factorizations with respect to all sets. This allowed us to see LTSs, 
modelled as $\{\ast\}$-pointed $\PP(A\times\arg)$-coalgebra, as an 
instance of our theory. This allowed us in particular to describe their trace semantics 
using our path category in Section \ref{subsec:trace}. This can be extended to tree 
automata as follows. Assume given a signature $\Sigma$, that is, a collection 
$(\Sigma_n)_{n\in \N}$ of disjoint sets. When $\sigma$ belongs to $\Sigma_n$, 
we say that $n$ is the \emph{arity of $\sigma$} or that $\sigma$ is a 
\emph{symbol of arity $n$}. A top-down non-deterministic tree automata as 
defined in \cite{comon07} is then the same as a $\{\ast\}$-pointed 
$\PP F$-coalgebra where $F$ is the polynomial functor 
$X \mapsto \coprod_{\sigma\in\Sigma_n} X^n$.  For this functor, $F^n(1)$ is the 
set of trees over $\Sigma\sqcup\{\ast(0)\}$ of depth at most $n+1$ such that a 
leaf is labelled by $\ast$ if and only if it is at depth $n+1$. Intuitively, elements of
$F^n(1)$ are partial runs of length $n$ that can possibly be extended.
Then, the trace semantics of a tree automata, seen as 
a pointed coalgebra, is given by the set of partial runs of the automata. In particular,
this contains the set of accepted finite trees as those partial runs without any $\ast$, 
and the set of accepted infinite trees, encoded as the sequence of their truncations 
of depth $n$, for every $n$.


In the following, we would like to extend this to other kinds of tree automata by 
allowing some symmetries. For example, in a tree, we may not care about the 
order of the children. This boils down to quotient the set $X^n$ of $n$-tuples,
 by some permutations of the indices. This can be done generally given a subgroup
  $G$ of the permutation group $\mathfrak{S}_n$ on $n$ elements by defining 
  $X^n/G$ as the quotient of $X^n$ under the equivalence relation: 
  $(x_1, \ldots, x_n) \equiv_G (y_1, \ldots, y_n)$ iff there is $\pi \in G$ such that for 
  all $i$, $x_i = y_{\pi(i)}$. Concretely, this means that we replace the polynomial 
  functor $F$ by a so-called \emph{analytic functor}:


\begin{notheorembrackets}
\begin{definition}[{\cite{joyal81,joyal86}}] An \emph{analytic} \Set-functor is
  a functor of the form $FX = \coprod_{\sigma\in \Sigma_n} X^{n}/G_\sigma$
  where for every $\sigma\in \Sigma_n$, we have a
  subgroup $G_\sigma$ of the permutation group $\mathfrak{S}_n$ on
  $n$ elements.
\end{definition}
\end{notheorembrackets}
\begin{example}
  Every polynomial functor is analytic. The bag-functor is analytic, with
  $\Sigma = (\{\ast\})_{n \in \N}$ has one operation symbol per arity and 
  $G_\sigma =
  \mathfrak{S}_{\ar(\sigma)}$ is the full permutation group on
  $\mathsf{ar}(\sigma)$ elements. It is the archetype of an analytic functor, in the sense 
  that for every analytic functor $F\colon \Set\to\Set$, there is a natural
  transformation into the bag functor $\alpha\colon F\to \Bag$. If $F$ is given
  by $\Sigma$ and $G_\sigma$ as above, then $\alpha_X$ is given by
  \[
    \textstyle
    FX = \coprod_{\sigma\in \Sigma_n} X^n/G_\sigma
    ~\twoheadrightarrow~
    \coprod_{\sigma\in \Sigma_n} X^n/\mathfrak{S}_n
    ~\rightarrow~
    \coprod_{n\in \N} X^n/\mathfrak{S}_n
    = \Bag X.
  \]
\end{example}
\begin{proposition}
\label{prop:analyticPrecise}
For an analytic \Set-functor $F$, the following are equivalent
~~(1) a map $f\colon X\to FY$ is $F$-precise,
~~(2) $\alpha_Y\cdot f$ is $\Bag$-precise,
~~(3) every element of $Y$ appears precisely once in the definition of $f$, i.e.~for every $y \in Y$, there is exactly one $x$ in $X$, such that $f(x)$ is the equivalence class of a tuple $(y_1, \ldots, y_n)$ where there is an index $i$, such that $y_i = y$; and furthermore this index is unique. 
	So every analytic functor 
  has precise factorizations w.r.t. $\Set$.
\end{proposition}

\subsection{Nominal Sets: Regular 
Nondeterministic Nominal Automata} \label{sec:RNNA}
We derive an open map situation from the coalgebraic situation for
\emph{regular nondeterministic nominal automata} (\emph{RNNAs})
\cite{schroder17}. They are an extension of
automata to accept \emph{words with binders}, consisting of literals
$a\in \A$ and binders $\barbind{a}$ for $a \in \A$; the latter is
counted as  
length 1.  An example of such
a word of length 4 is $a\barbind{c}bc$, where the last $c$ is bound by
$\barbind{c}$. The order of binders makes difference:
$\barbind{a}\barbind{b}ab \neq \barbind{a}\barbind{b}ba$.
RNNAs are coalgebraically represented in
the category of nominal sets \cite{gabbay99}, a formalism about atoms 
(e.g. variables) that sit in more
complex structures (e.g. lambda terms), and gives a notion of \textit{binding}.
Because the choice principles \eqref{axJoin} and \eqref{axChoice} are not satisfied by every nominal sets, we instead use the class of \emph{strong nominal sets} for the precise factorization (Definition \ref{defPreciseFactor}).

\begin{notheorembrackets}%
\begin{definition}[\cite{gabbay99,pitts13}] Fix a countably infinite set $\A$, called the set of \emph{atoms}.
For the group $\perms(\A)$ of finite permutations on the set $\A$, a
  \emph{group action} $(X,\cdot)$ is a set $X$ together with a group homomorphism
  $\cdot\colon\perms(\A)\to \perms(X)$, written in infix notation. An element $x\in
  X$ is \emph{supported by} $S\subseteq \A$, if for all $\pi\in\perms(\A)$ with
  $\pi(a) =a \ \forall a\in S$ we have $\pi\cdot x = x$.
   A \emph{nominal set} is a group action for $\perms(\A)$ such that every $x\in
  X$ is finitely supported, i.e.~supported by a finite $S\subseteq \A$. A map
  $f\colon (X,\cdot)\to (Y,\star)$ is \emph{equivariant} if for all $x\in X$ and
  $\pi\in \perms(\A)$ we have $f(\pi\cdot x) = \pi\star f(x)$. The category of
  nominal sets and equivariant maps is denoted by $\Nom$. 
  A nominal set $(X,\cdot)$ is called \emph{strong} if for all $x\in X$ and
  $\pi\in \perms(\A)$ with $\pi\cdot x = x$ we have $\pi(a) = a$ for all $a\in \supp(x)$.
\end{definition}
\end{notheorembrackets}
Intuitively, the support of an element is the set of free literals. An
equivariant map can forget some of the support of an element, but can never
introduce new atoms, i.e.~$\supp(f(x)) \subseteq \supp(x)$. The intuition behind
strong nominal sets is that all atoms appear in a fixed order, that is, $\A^n$ is
strong, but $\powf(\A)$ (the finite powerset) is not. We set $\S$
to be the class of strong nominal sets:
\begin{example} \label{ex:whyClassS}
  The $\Nom$-functor of unordered pairs admits precise factorizations w.r.t.~strong
  nominal sets, but not w.r.t.~all nominal sets.
\end{example}
In the application, we fix the set $I = \A^{\#n}$ of distinct $n$-tuples of atoms
($n\ge 0$) as the pointing. The hom-sets $\Nom(X,\powufs Y)$ are ordered
point-wise.
\begin{proposition}
\label{prop:powersetNom}
Uniformly finitely supported powerset $\powufs(X) = \{Y\subseteq X \mid  \bigcup_{y\in
  Y}\supp(y)\text{ finite}\}$ satisfies \textup{(Ax2-5)} w.r.t. $\S$ the class of
strong nominal sets.\footnote{There are two variants of powersets discussed
  in \cite{schroder17}. The finite powerset
  \powf\xspace also fulfils the axioms. However, \emph{finitely supported}
  powerset $\mathcal{P}_{\mathsf{fs}}$ does not fulfil \eqref{axChoice}.}
\end{proposition}
As for $F$, we study an LTS-like functor, extended with the \emph{binding functor} \cite{gabbay99}\rlap{:}
\begin{definition}
  For a nominal set $X$, define the $\alpha$-equivalence relation
  $\sim_\alpha$ on $\A\times X$ by:
  \(
    (a,x) \sim_\alpha (b,y)
    ~\Leftrightarrow~
    \exists c\in \A\setminus\supp(x)\setminus\supp(y)\text{ with }(a\,c)\cdot x = (b\,c)\cdot y.
  \)
  Denote the quotient by $[\A]X := \A\times X / \mathord{\sim}_\alpha$. The
  assignment $X\mapsto [\A]X$ extends to a functor, called the
  \emph{binding functor} $[\A]\colon \Nom\to\Nom$.
\end{definition}
RNNA are precisely $\powufs F$-coalgebras
for $FX = \{\checkmark\}+[\A]X+\A\times X$ \cite{schroder17}.
In this paper we additionally consider initial states for RNNAs.
\begin{proposition}
\label{prop:nomPrecise}
The binding functor $[\A]$ admits precise factorizations w.r.t. strong nominal
sets and so does $FX=\{\checkmark\}+[\A]X+\A\times X$.
\end{proposition}

An element in $\PathOrd(\A^{\#n}, F)$ may be regarded as a word with
binders under a context $\vec a\vdash w$, where
$\vec a\in\A^{\#n}$, all literals in $w$ are bound or in $\vec a$, and $w$ may end with $\checkmark$.
Moreover, two word-in-contexts $\vec a\vdash w$ and $\vec a'\vdash w'$ are
identified if their closures are $\alpha$-equivalent, that is,
$\barbind{a_1}\cdots\barbind{a_n}w=\barbind{a'_1}\cdots\barbind{a'_n}w'$.  The trace semantics of a
RNNA $T$ contains all the word-in-contexts corresponding to runs in
$T$. This trace semantics distinguishes whether words are concluded by
$\checkmark$.

\subsection{Subsuming arbitrary open morphism situations}
\label{subsec:Lasota}
Lasota~\cite{lasota02} provides a translation of a small path-category
$\bbP\hookrightarrow \M$ into a functor $\mathbb{F}\colon \Set^{\obj\bbP} \to
\Set^{\obj\bbP}$ defined by $\mathbb{F}\big(X_P\big)_{P}=(\prod_{Q\in\bbP}\big(\pow(X_Q))^{\bbP(P,Q)}\big)_{P\in \bbP}$.
So the hom-sets $\Set^{\obj\bbP}(X,\mathbb{F}Y)$ have a canonical order, namely
the point-wise inclusion. This admits a functor $\oper{Beh}$ from $\M$ to
$\mathbb{F}$-coalgebras and lax coalgebra homomorphisms, and Lasota shows that 
$f\in \M(X,Y)$ is $\bbP$-open iff $\oper{Beh}(f)$ is a coalgebra homomorphism.
In the following, we show that we can apply our framework to $\mathbb{F}$ by a
suitable decomposition $\mathbb{F} = TF$ and a suitable object $I$ for the
initial state pointing. As usual in open map papers, we require that
$\bbP$ and $\mathbb{M}$ have a common initial object $0_\bbP$.
Observe that we have $ \mathbb{F} = T\cdot F $ where
\[
  T(X_P)_{P\in \bbP}
  = \big(\pow (X_P)\big)_{P\in \bbP}
  \quad
  \text{and}
  \quad
  F(X_P)_{P\in \bbP}
  = \big(
  \smash{\textstyle\coprod_{Q\in \bbP}}
  {\bbP(P,Q)}\times X_Q\big)_{P\in \bbP}.
\]
Lasota considers coalgebras without pointing, but one indeed has a canonical
pointing as follows.
For $P\in \bbP$, define the characteristic family $\chi^P\in \Set^{\obj\bbP}$ by
$\chi^P_Q = 1$ if $P=Q$
and $\chi^P_Q = \emptyset$ if $P\neq Q$.
With this, we fix the pointing $I = \chi^{0_\bbP}$.
\begin{proposition}
\label{prop:LasotaAxioms}
$T$, $F$ and $I$ satisfy the axioms from \autoref{table:axioms}, with $\S = \Set^{\obj\bbP}$.
\end{proposition}
The path category in $\LCoalg(I,TF)$ from our theory can be described as follows.

\begin{proposition}
\label{prop:LasotaPrecise}
An object of $\Path(I,F)$ is a sequence of composable $\bbP$-mor-\\phisms $0_\bbP \xrightarrow{m_1} P_1
  \xrightarrow{m_2} P_2\cdots \xrightarrow{m_n} P_n$.
\end{proposition}


\section{Conclusions and Further work}
\label{sec:furtherWork}
We proved that coalgebra homomorphisms for systems with
non-deterministic branching can be seen as open maps for a canonical
path-category, constructed from the computation type $F$. This limitation to
non-deterministic systems is unsurprising: as we have proved in Section
\ref{subsec:Lasota} on Lasota's work \cite{lasota02}, every open map situation
can been encoded as a coalgebra situation with a powerset-like functor, so with
non-deterministic branching. As a future work, we would like to extend this
theory of path-categories to coalgebras for further kinds of
branching, especially probabilistic and weighted. This will require (1) to
adapt open maps to allow those kinds of branching (2) adapt the axioms from
\autoref{table:axioms}, by replacing the ``+1'' part of
\eqref{axFactor} to something depending on the branching type.
\newpage

\bibliography{refs}
\bibliographystyle{splncs04}
\newpage
\begin{appendix}
\tikzsetfigurename{appendix-}
\section{Omitted Proofs}

\subsection*{Proof of \autoref{rem:PreciseWeakPullbacks}}

\begin{lemma} \label{preciseWeakPullback}
  Assuming that $F$ preserves weak pullbacks, a morphism $s\colon S\to FR$
  is $F$-precise iff for all $f,h$ the following implication holds:
  \[
    \begin{tikzcd}
      S
      \arrow{r}{f}
      \arrow{d}[swap]{s}
      & FC
      \arrow{dl}{Fh}
      \\
      FR
    \end{tikzcd}
    \quad\overset{\exists d\,}{\Longrightarrow}\quad
    \begin{tikzcd}
      S
      \arrow{r}{f}
      \arrow{d}[swap]{s}
      & FC
      \\
      FR
      \arrow{ur}[swap]{Fd}
    \end{tikzcd}
    \&
    ~~
    h\cdot d = \id_R
  \]
\end{lemma}
\begin{proof}
  Sufficiency is clear, because one can fix $D=R$ and $g=\id_R$ in the
  definition of $F$-precise. For necessity, consider $s,f,g,h$ as in the definition of
  $F$-precise. The pullback of $g$ along $h$ is weakly preserved by $F$, and so
  we have the following commuting diagram:
  \[
    \begin{tikzcd}
      |[yshift=8mm]|
      S
      \arrow[dashed]{r}[swap,near end]{\exists f'}
      \arrow{d}[swap]{s}
      \arrow[bend left=10]{rr}{f}
      & FP
      \arrow{r}[swap]{F\pi_1}
      \arrow{dl}{F\pi_2}
      \pullbackangle{-45}
      & FC
      \arrow{d}{Fh}
      \\
      FR
      \arrow{rr}[near end]{Fg}
      && FD
    \end{tikzcd}
  \]
  Hence, $s$ induces some $d\colon R\to P$ with $Fd\cdot s= f'$ and $\pi_2\cdot
  d = \id_R$. The witness that $s$ is $F$-precise is $\pi_1\cdot d\colon R\to
  C$, because $F(\pi_1\cdot d) \cdot s = f$ and $h\cdot \pi_1\cdot d = g$.\qed
\end{proof}

\subsection*{Proof of \autoref{remPreciseFactor}}

  Apply that $f_2$ is $F$-precise and one obtains a lifting $d\colon Y_2\to Y_1$ with
  $Fd\cdot f_2 = f_1$ and $h\cdot d = \id_{Y_2}$, i.e.~$d$ has a left inverse, $h$.
  Additionally, since $f_1$ is $F$-precise, $Fd\cdot f_2 = f_1$ induces some $d'\colon
  Y_1\to Y_2$ with $f_2=Fd'\cdot f_1$ and $d\cdot d' = \id_{Y_1}$. Since $d$ has both a
  left and right inverse, it is an isomorphism, and so is its left-inverse $h$.
  
  For the second part, we a have a precise factorisation of $f$, that is, a $F$-precise 
  morphism $\map{f'}{S}{Y'}$ with $Y' \in \S$, and a morphism $\map{h}{Y'}{Y}$ such that 
  $f = Fh\cdot f'$. By the previous point, $h$ is an isomorphism, and since $\S$ is closed 
  under isomorphisms, $Y \in \S$.

\subsection*{Proof of \autoref{propFactorizationClosure}}
\begin{enumerate}

	\item Fix $Z$ an object of $\C$, and assume given a morphism 
	$\map{f}{S}{Z = \C_ZY}$, with $S \in \S$. Since a morphism of the form 
	$\map{f'}{S}{F0}$ is always $F$-precise, then $\map{f' = f}{S}{Z = \C_Z0}$ 
	is $\C_Z$-precise. Taking $h$ as the unique morphism from $0$ to $Y$, 
	$Fh\cdot f' = \text{id}_Z\cdot f = f$.
	
	\item Let us start with the following lemma:
	
	\begin{lemma}
	If $\map{f}{S}{FR}$ is $F$-precise and $\map{f'}{R}{F'Q}$ is $F'$-precise, then 
	$Ff'\cdot f$ is $F\cdot F'$-precise.
	\end{lemma}
	
	\begin{proof}
	Assume given the following situation:
  \[
    \begin{tikzcd}
      S
      \arrow{r}{k}
      \arrow{d}[swap]{f}
      & FF'C
      \arrow{dd}{FF'u}
      \\
      FR
      \arrow{d}[swap]{Ff'}
      \\
      FF'Q
      \arrow{r}{FF'v}
      & FF'W
    \end{tikzcd}
  \]
	
	In particular, we have $(FF'v\cdot Ff')\cdot f = FF'u\cdot k$. Since $f$ is $F$-precise,
	there is $\map{d}{R}{F'C}$ such that 
	$$k = Fd\cdot f$$
	$$F'v\cdot f' = F'u\cdot d$$
	Since $f'$ is $F'$-precise, there is $\map{d'}{Q}{C}$ such that
	$$d = F'd'\cdot f'$$
	$$v = u \cdot d'$$
	This implies that $k = FF'd'\cdot(Ff'\cdot f)$.\qed
	\end{proof}
	
	Let us prove this point now. We start with a morphism $\map{f}{S}{FF'Y}$, with 
	$S \in \S$. Since $F$ admits precise factorisations w.r.t. $\S$, there are 
	$\map{f'}{S}{FY'}$ $F$-precise and $\map{h}{Y'}{F'Y}$ with $Y'\in \S$, and 
	$Fh\cdot f' = f$. Now, since $F'$ admits precise factorisations w.r.t. $\S$, there are 
	$\map{f''}{Y'}{F'Y''}$ $F'$-precise and $\map{h'}{Y''}{Y}$ such that $F'h'\cdot f'' = h$.
	Consequently, $f = FF'h'\cdot(Ff''\cdot f')$ and $Ff''\cdot f'$ is $FF'$-precise by the 
	previous lemma.
	
	\item Let us start with the following lemma:
	
	\begin{lemma}
  \label{FxGprecise}
  For a family of functors $F_i\colon \C\to \C$, $i \in I$, and $F_i$-precise morphisms
  $f_i\colon X\to F_iY_I$, $g\colon X\to GY_G$, then we have a 
  $\prod\limits_{i\in I} F_i$-precise morphism
  \[
    \fpair{F_i\inj_i \cdot f_i}\colon X\to \prod\limits_{i\in I} F_i(\coprod\limits_{j\in I} Y_j).
  \]
\end{lemma}

\begin{proof}
  Consider a square
  \[
    \begin{tikzcd}
      X
      \arrow{r}{\fpair{v_i}_{i\in I}}
      \arrow{d}[swap]{\fpair{f_i}_{i\in I}}
      &
      \prod\limits_{i\in I} F_iW
      \arrow{dd}{[F_iu]_{i\in I}}
      \\
      \prod\limits_{i\in I} F_iY_i
      \arrow{d}[swap]{[F_i\inj_i]_{i\in I}}
      &
      \\
      \prod\limits_{i\in I} F_i(\coprod\limits_{j\in I} Y_j)
      \arrow{r}{[F_ih]_{i\in I}}
      & 
      \prod\limits_{i\in I} F_iZ
    \end{tikzcd}
    \Longrightarrow
    \begin{tikzcd}
      X
      \arrow{r}{v_i}
      \arrow{d}[swap]{f_i}
      &
      FW
      \arrow{dd}{F_iu}
      \\
      F_iY_i
      \arrow{d}[swap]{F_i\inj_i}
      \\
      F_i(\coprod\limits_{j\in I} Y_j)
      \arrow{r}{F_ih}
      & 
      F_iZ
    \end{tikzcd}
  \]
  Since $f_i$ is $F_i$-precise, we obtain some $d_i\colon Y_i\to W$ with $F_id_i\cdot
  f_i = v_i$ and $u\cdot d_i = h\cdot \inj_i$. We have 
  $[d_j]_{j\in I}\colon \coprod\limits_{j\in I} Y_j\to W$ with
  \[
    (\prod\limits_{i\in I}F_i[d_j]_{j\in I}) \cdot \fpair{F_i\inj_i\cdot f_i}_{i\in I}
    = \fpair{Fd_i\cdot f_i}_{i\in I}
    = \fpair{v_i}_{i \in I}
  \]
  and
  \[
    u\cdot [d_i]_{i\in I} = [u\cdot d_i]_{i\in I} = [h\cdot \inj_i]_{i\in I} = h.
  \]
  \qed
\end{proof}
	
	Let us prove this point now. Consider $\fpair{f_i}_{i\in I}\colon X\to \prod\limits_{i\in I} F_iY$ 
	and consider the $F_i$-precise factorizations:
  \[
    \begin{tikzcd}
      X \arrow{r}{f_i}
      \arrow[dashed]{dr}[swap]{f'_i}
      & F_iY
      \\
      & F_iY_i
      \arrow{u}[swap]{F_ih_i}
    \end{tikzcd}
  \]
  By \autoref{FxGprecise} we have that $\fpair{F_i\inj_i\cdot f'_i}\colon 
  X\to \prod\limits_{i\in I} F_i(\coprod\limits_{j\in I} Y_j)$ is 
  $\prod\limits_{i\in I} F_i$-precise and it is the
  $\prod\limits_{i\in I} F_i$-precise factorization of $\fpair{f_i}_{i\in I}$:
  \[
    (\prod\limits_{i\in I} F_i)[h_i]_{i\in I}\cdot \fpair{F_i\inj_i\cdot f'_i}_{i\in I}
    = \fpair{F_ih_i\cdot f'_i}_{i\in I} = \fpair{f_i}_{i\in I}.
  \]
	
	\item By $I$-extensive we mean `extensive' if $I$ is finite and `infinitary
    extensive' if $I$ is infinite. In any case, $\C$ is $I$-extensive if for
    for all families $(X_i)_{i\in I},(Y_i)_{i\in I},(g_i)_{i\in I},(x_i)_{i\in
      I}$ and $h$ with
    \[
      \begin{tikzcd}
        X_i
        \arrow{r}{x_i}
        \arrow{d}[swap]{g_i}
        & P
        \arrow{d}{h}
        \\
        Y_i\arrow{r}{\inj_i}
        & \coprod_{j\in I} Y_j
      \end{tikzcd}
      \qquad\forall{i\in I}
    \]
    the following equivalence holds:
    \[
      \begin{array}{c}
      (x_i\colon X_i\to P)_{i\in I}
      \text{ is a coproduct}
      \\\Longleftrightarrow\\
      \text{ for all }i\in I\text{ the above square is a pullback}
      \end{array}
    \]
    Let us start with the following lemma:
	
	\begin{lemma}
  \label{coproductPrecise}
  Given functors $F_i\colon \C\to\C$, $i\in I$, on an $I$-extensive category, 
  and morphisms
  $f_i\colon A_i\to F_iX_i$, $i\in I$. Then \( \coprod_{i\in I}(F_i\inj_i\cdot
  f_i) \colon \coprod_{i\in I}A_i\to \big(\coprod_{i\in
    I}F_i\big)\big(\coprod_{i\in I}X_i\big) \) is $\coprod_{i\in I}F_i$-precise
  if $f_i$ is $F_i$-precise for every $i\in I$.
\end{lemma}

\begin{proof}
  Consider a commutative square
  \[
    \begin{tikzcd}[column sep=5mm]
      \coprod_{i\in I}A_i
      \arrow{d}[swap]{\coprod_{i\in I}f_i}
      \arrow{r}{\coprod_{i\in I} a_i}
      &[30mm] \coprod_{i\in I}F_iC
      \arrow{dd}{\coprod_{i\in I}F_ic}
      \\
      \coprod_{i\in I}(F_iX_i)
      \arrow{d}[swap]{\coprod_{i\in I}F_i\inj_i}
      \arrow{dr}{\coprod_{i\in I}F_ig_i}
      \\
      \big(\coprod_{i\in I}F_i\big)\big(\coprod_{i\in I}X_i\big)
      \arrow{r}[swap]{\big(\coprod_{i\in I}F_i\big)[F_ig_i]_{i\in I}}
      & \coprod_{i\in I}F_iD
      \\
      \coprod_{i\in I}F_i(\coprod_{j\in I} X_i)
      \descto[sloped]{u}{=}
    \end{tikzcd}
  \]
  Note that by the extensivity of $\C$ we can assume that the top morphism
  $\coprod_{i\in I}A_i\to \coprod_{i\in I}F_iC$ is indeed a coproduct of
  morphisms. Hence we have $F_i g_i\cdot f_i = F_ic\cdot
  a_i$, for every $i\in I$. Since $f_i$ is $F_i$-precise, this induces some
  $d_i\colon X_i\to C$ with $F_id_i\cdot f_i = a_i$ and $c\cdot g_i = d_i$.
  In total, we have $[d_i]_{i\in I}\colon \coprod_{i\in I} X_i\to C$ with
  $c\cdot [d_i]_{i\in I} = [g_i]_{i\in I}$ and
  \[
    \big(\coprod_{i\in I}F_i\big)[d_i]_{i\in I}
    \cdot \coprod_{i\in I}(F_i\inj_i\cdot f_i)
    = \coprod_{i\in I}(F_id_i\cdot f_i)
    = \coprod_{i\in I} a_i
    \tag*{\qed}
  \]
\end{proof}
	
	Let us prove this point now. Given a morphism $h\colon X\to (\coprod_{i\in I}F_i)(Y)$, construct the
  pullbacks in the $I$-extensive category $\C$:
  \[
    \begin{tikzcd}
      X_i
      \arrow{r}{\inj_i}
      \arrow{d}[swap]{f_i}
      \pullbackangle{-45}
      & X
      \arrow{d}{h}
      \\
      F_iY
      \arrow{r}{\inj_i}
      & (\coprod_{i\in I}F_i)(Y)
    \end{tikzcd}
    \text{ for every }i\in I.
  \]
  So we have $X = \coprod_{i\in I}X_i$ with the coproduct injections as in the
  top row of the above pullback diagram, and in particular $h = [\inj_i\cdot
  f_i]_{i\in I} = \coprod_{i\in I} f_i$.
  Let $f_i'$ be the $F_i$-precise morphisms through
  which $f_i$ factors for every $i\in I$:
  \[
    \forall i\in I\quad
    \begin{tikzcd}
      X_i
      \arrow[dashed]{r}{\exists f_i'}
      \arrow{dr}[swap]{f_i}
      & FY_i
      \arrow{d}{Fy_i}
      \\
      & FY
    \end{tikzcd}
    ~\Longrightarrow
    \hspace{-4mm}
    \begin{tikzcd}[column sep=18mm]
      \coprod_{i\in I} X_i
      \arrow{r}[yshift=2mm]{\coprod_{i\in I}F\inj_i\cdot f_i'}
      \arrow{dr}[swap]{h = \coprod_{i\in I}f_i}
      & (\coprod_{i\in I} F_i)(\coprod_{i\in I}Y_i)
      \arrow{d}{(\coprod_{i\in I} F_i)(\coprod_{i\in I}y_i)}
      \\
      & (\coprod_{i\in I} F_i)(Y)
    \end{tikzcd}
  \]
  By \autoref{coproductPrecise}, $\coprod_{i\in I}F\inj_i\cdot f_i'$ is
  $\coprod_{i\in I}F_i$-precise.
  \\[2mm]
	
	\item We first show that $\eta_X\colon X\to R(LX)$ is $R$-precise for
  every $X\in \C$. Since right adjoints preserve limits, $R$ preserves (weak)
  pullbacks and it suffices to check \autoref{preciseWeakPullback}/\autoref{rem:PreciseWeakPullbacks}. For any
  suitable commutative triangle we have
  \[
    \begin{tikzcd}
      X
      \arrow{r}{g}
      \arrow{d}[swap]{\eta_X}
      & RZ
      \arrow{dl}{Rh}
      \\
      RLX
    \end{tikzcd}
  \]
  Since $L\dashv R$, there exists a unique $g'$ with
  \[
    \begin{tikzcd}
      X
      \arrow{r}{g}
      \arrow{d}[swap]{\eta_{X}}
      & RZ
      \\
      RLX
      \arrow{ur}[swap]{Rg'}
    \end{tikzcd}
    \text{ and }
    ~
    \begin{tikzcd}
      LX
      \arrow{r}{g'}
      \arrow{d}[swap]{\id_{LX}}
      & Z
      \arrow{dl}{h}
      \\
      LX
    \end{tikzcd}
  \]
  by the universal mapping property of $L\dashv R$ and by the naturality of the
  isomorphism $\C(LX,Z)\cong \C(X,RZ)$ respectively. Now, we can prove this point:
  \begin{itemize}
  \item For necessity and $X\in \S$, we have that $\eta_X\colon X\to RLX$ is
    $R$-precise by the previous argument, and hence $LX\in \S$ by
    \autoref{remPreciseFactor}.
  \item For sufficiency, consider a morphism $f\colon X\to RY$. The adjunction
    induces a unique $f'\colon LX\to Y$ with $Rf'\cdot \eta_X = f$, where
    $\eta_X\colon X\to RLX$ is $R$-precise as shown above. Since $L$ preserves
    objects in $\S$, we have that $LX\in \S$.
  \end{itemize}
\end{enumerate}

\subsection*{Proof of \autoref{powersetNotPrecise}}
Given an $\pow$-precise $f\colon X\to \pow Y$, define $f'\colon X\to \pow (Y+Y)$
by $f'(x) = \{\inl(y), \inr(y)\mid y\in f(x)\}$, and so $\pow[\id_Y,\id_Y]\cdot
f' = f$. Hence, we have some $d\colon Y\to Y+Y$ with $[\id_Y,\id_Y]\cdot d =
\id _Y$ and $f' = \pow d\cdot f$. The first equation implies that for every $y\in
Y$, $d(y)$ is $\inl(y)$ or $\inr(y)$. So for $x\in X$ and for every $y\in Y$, either $\inl(y)$
or $\inr(y)$ is in $d[f(x)] = (\pow d\cdot f)(x) = f'(x)$, which is a
contradiction unless $X$ or $Y$ is empty. Hence, $f(x) = \emptyset$ for all
$x\in X$.

\subsection*{Proof of \autoref{remarkPathCategory}}
  The proof is by induction on the length.
  Since $P_0 = I = Q_0$, $\phi_0\colon P_0\to Q_0$ is an isomorphism.
  Assume that $\phi_k\colon P_k\to Q_k$ is isomorphic for $k < n$.
  Since $q_k\cdot \phi_k$ and $p_k$ are $F+1$-precise, we have that $\phi_{k+1}$
  is isomorphic by \autoref{remPreciseFactor}.
  
  The second part of the remark is a consequence of the second part of 
  \autoref{remPreciseFactor}.

\subsection*{Proof of \autoref{prop:PathToPathOrd}}
\begin{enumerate}
	\item
  \begin{itemize}
  \item \textbf{Functoriality.} consider a path morphism $\vec\phi_{n+1}\colon (\vec
  P_{n+1}, \vec p_n) \to (\vec Q_{m+1}, \vec q_m)$ and use that all components
  are isos~(\autoref{remPreciseFactor})
  \[
    \begin{tikzcd}
      I= Q_0
      \arrow[equals]{d}
      \arrow{r}{q_{0}}
      & FQ_1
      \arrow{d}[swap]{F\phi_1^{-1}}
      \arrow[dotted,-]{r}
      \arrow[shorten <= 4mm]{r}
      & F^{n}Q_{n}
      \arrow{d}[swap]{F^{n}\phi_{n+1}^{-1}}
      \arrow[dotted,-]{r}
      \arrow[shorten <= 4mm]{r}
      & F^{m}Q_{m}
      \arrow{r}{F^{m}!}
      & F^{m}1
      \arrow{dl}{F^{n}!}
      \\
      I= P_0
      \arrow{r}{p_{0}}
      & FP_1
      \arrow[dotted,-]{r}
      \arrow[shorten <= 4mm]{r}
      & F^{n}P_{n}
      \arrow{r}{F^{n}!}
      & F^{n}1
    \end{tikzcd}
  \]
  Hence $\Comp(\vec P_{n+1}, \vec p_n) \le \Comp(\vec Q_{m+1}, \vec q_m)$.
  \item \textbf{Fullness.} Since $\PathOrd(I,F)$ is partially ordered, it
    sufficies to show that whenenver $\PathOrd(I,F)(\Comp(\vec P_{n+1}, \vec
    p_n),\Comp(\vec Q_{m+1}, \vec q_m))$ is non-empty, then so is
    $\Path(I,F)(\vec P_{n+1}, \vec p_n),(\vec Q_{m+1}, \vec q_m)$. We have $n\le
    m$ and construct a morphism $\vec\phi_{n+1}\colon (\vec P_{n+1},\vec p_n) \to (\vec
    Q_{m+1},\vec q_m)$ by induction. There is nothing to do in the base case
    $\phi_0 := \id_I$. For the step, assume some $\phi_k\colon P_k\to Q_k$
    making the following diagram commute:
    \[
      \begin{tikzcd}
        P_k
        \arrow{r}{\phi_k}
        \arrow{d}{p_k}
        & Q_k
        \arrow{r}{q_k}
        & FQ_{k+1}
        \arrow[dotted,-]{r}
        \arrow[shorten <= 4mm]{r}
        & F^{n-k} Q_n
        \arrow{d}{F^{n-k}!}
        \\
        FP_{k+1}
        \arrow[dotted,-]{r}
        \arrow[shorten <= 4mm]{r}
        & F^{n-k} P_n
        \arrow{rr}{F^{n-k}!}
        && F^{n-k} 1
      \end{tikzcd}
    \]
    Since $p_k$ is $F$-precise, we obtain a morphism $\phi_{k+1}\colon
    P_{k+1}\to Q_{k+1}$ with
    \[
      \begin{tikzcd}
        P_k
        \arrow{r}{\phi_k}
        \arrow{d}[swap]{p_k}
        & Q_k
        \arrow{d}{q_k}
        \\
        FP_{k+1}
        \arrow{r}{F\phi_{k+1}}
        & FQ_{k+1}
      \end{tikzcd}
      \text{and}
      \begin{tikzcd}
        P_{k+1}
        \arrow{r}{\phi_{k+1}}
        \arrow[dotted,-]{d}
        \arrow[shorten <= 2mm]{d}
        & Q_{k+1}
        \arrow[dotted,-]{r}
        \arrow[shorten <= 4mm]{r}
        & F^{n-k-1} Q_n
        \arrow{d}{F^{n-k-1}!}
        \\
        F^{n-k-1} P_n
        \arrow{rr}{F^{n-k-1}!}
        && F^{n-k-1} 1.
      \end{tikzcd}
    \]
    So, $\vec\phi_{k+1}$ is indeed a morphism in $\Path(I,F)$.
  \end{itemize}
  \item For every $u\colon I\to F^n 1$ we can
  define a path $(\vec P_{n+1},\vec p_n)$ inductively, starting with $P_0 := I
  \in \S$. Given some $\C$-morphism $u'\colon P_k\to F^{n-k} 1$, $k\le n$,
  $P_k\in \S$, we are done with $r = \id_1$ if $k=n$. If $k < n$ consider the
  $F$-precise factorization
    \[
      \begin{tikzcd}
        P_k \arrow{dr}[swap]{r}
        \arrow[dashed]{r}{p_k}
        & FP_{k+1}
        \arrow{d}{Fr'}
        \\
        & FF^{n-k-1}1
      \end{tikzcd}
    \]
    providing $P_{k+1}\in \S$ and $p_k\colon P_k\to FP_{k+1}$ and some $u''\colon
    P_{k+1}\to F^{n-k-1}1$. This defines the families $\vec P_{n+1}$ and $\vec
    p_n$ with $\Comp(\vec P_{n+1},\vec p_n) = u$.
  \item By \autoref{uniqueFactorization}.
  \end{enumerate}

\subsection*{Proof of \autoref{lem:stepwiseFromPath}}
For all $k < n$, the following diagram commutes:
  \[
    \begin{tikzcd}[column sep=8mm]
      P_k
      \arrow{r}{\inj_k}
      \arrow{d}[swap]{p_k}
      &[7mm] \coprod_{j\le n} P_j
      \arrow{d}[pos=0.4,description]{\smash{\big[[\inl\cdot F\inj_{j+1}\cdot
          p_j]_{j<n},\inr\cdot!\big]}
        \mathrlap{\phantom{p}}
      }
      \arrow[bend left=20]{dr}[above right]{=:\, \rho}
      \\[5mm]
      FP_{k+1} + 1
      \arrow{d}[swap]{Fx_{k+1}+1}
      \arrow{r}{F\inj_{k+1}+1}
      & F\coprod_{j\le n}P_{j} + 1
      \arrow{r}{[\eta,\bot]}
      \arrow{d}[swap]{F[x_j]_{j\le n}+1}
      & TF\coprod_{j\le n}P_{j+1}
      \arrow{d}{TF[x_j]_{j\le n}}
      \\
      FX+1
      \arrow[equals]{r}
      & FX+1
      \arrow{r}[swap]{[\eta,\bot]}
      & TFX
    \end{tikzcd}
  \]
  For sufficiency, assume that $[x_k]_{k\le n}$ is a morphism in $\LCoalg(I,TF)$. Hence
  for all $k < n$:
  \[
    [\eta,\bot]\cdot (Fx_{k+1}+1)\cdot p_k
    = TF[x_j]_{j\le n}\cdot \rho\cdot \inj_k
    \sqsubseteq \xi\cdot [x_j]_{j\le n}\cdot \inj_k
    = \xi\cdot x_k.
  \]
  For necessity, we have for all $k< n$:
  \[
    \begin{tikzcd}
      P_k
      \arrow[shift right=1]{d}[swap]{p_k}
      \arrow{r}{\inj_k}
      \arrow[rounded corners,to path={
        ([xshift=2pt]\tikztostart.south)
        -- ([xshift=2pt,yshift=-6mm]\tikztostart.base)
        -- ([xshift=-2pt,yshift=-6mm]\tikztotarget.base) \tikztonodes
        -- ([xshift=-2pt]\tikztotarget.south)
      }]{rr}{x_k}[alias=arrowcenter,anchor=center]{}
      \arrow[rounded corners,to path={
        -- ([xshift=-8mm]\tikztostart.west)
        |- ([yshift=-4mm]\tikztotarget.base) \tikztonodes
        -- (\tikztotarget)
      }]{drr}[pos=0.8,below]{TF[x_j]_{j\le n}\cdot \rho \cdot \inj_k}
      &
      \coprod_{j\le n}P_j
      \arrow{r}{[x_j]_{j\le n}}
      & X
      \arrow[shift left=1]{d}{\xi}
      \\[4mm]
      FP_{k+1}+1
      \arrow{r}{Fp_{k+1}+1}
      & FX+1
      \arrow{r}{[\eta,\bot]}
      \arrow[from=arrowcenter,draw=none]{}[anchor=center]{\sqsubseteq}
      & TFX
    \end{tikzcd}
  \]
  and
  \[
    \begin{tikzcd}[column sep=3mm]
      P_n
      \arrow{d}[swap]{\inr\cdot !}
      \arrow{dr}{!}
      \arrow{r}{\inj_n}
      \arrow[shiftarr={xshift=-15mm}]{dd}[sloped,above,rotate=180]{\rho\cdot \inj_n}
      & \coprod_{j\le n} P_j
      \arrow{r}{[x_j]_{j\le n}}
      & X
      \arrow{dd}{\xi}
      \\
      F\coprod_{j\le n}P_{j} + 1
      \arrow{d}[swap]{[\eta,\bot]}
      & 1
      \arrow{l}[pos=0.6,swap]{\inr}
      \arrow{dr}[swap]{\bot}
      \descto[pos=0.4,sloped]{ur}{$\sqsubseteq$}
      \arrow{dl}{\bot}
      \\
      TF\coprod_{j\le n}P_{j}
      \arrow{rr}[swap]{TF[x_j]_{j\le n}}
      &
      & TFX
    \end{tikzcd}
  \]
  Since the family of coproduct injections $\inj_k$ is jointly epic we have by
 \eqref{axJointlyepic} that $[x_k]_{k\le n}$ is indeed a
  weak homomorphism.

\subsection*{Proof of \autoref{thm:homIsOpen}}
As usual in open map proofs, this direction of the characterization theorem
is shown by considering only \Path-morphisms of length difference 1.
Indeed every \Path-morphism $\phi$ from a path of length $n$ to a path of length $m
\ge n+2$ can be expressed as a composition of
\Path-morphisms of smaller length difference $\phi = \psi\cdot \theta$.
When applying the openness of a morphism $h$ to $\theta$, we obtain a diagonal $d\colon
\dom(\psi)\to \dom(h)$ which allows applying the openness of $h$ to $\psi$,
yielding the desired lifting for $\phi$.
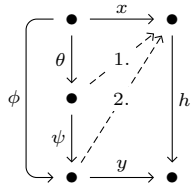
\[
  \begin{tikzcd}
    \bullet
    \arrow[shiftarr={xshift=-6mm}]{dd}[swap]{\phi}
    \arrow{d}[swap]{\theta}
    \arrow{r}{x}
    & \bullet
    \arrow{dd}{h}
    \\
    \bullet
    \arrow{d}[swap]{\psi}
    \arrow[dashed]{ur}[description]{1.}
    &
    \\
    \bullet
    \arrow{r}{y}
    \arrow[dashed]{uur}[description]{2.}
    & \bullet{}
  \end{tikzcd}
\]
Consider a $\Path(I,F+1)$-morphism
  $\vec \phi_{n+1}\colon (\vec P_{n+1}, \vec p_n) \to (\vec Q_{n+2},\vec
  q_{n+1})$ and a commuting square in $\LCoalg(I,TF)$ where $h$ is a homomorphism:
  \[
    \begin{tikzcd}[column sep=12mm]
      J(\vec P_{n+1}, \vec p_n)
      \arrow{r}{[x_k]_{k\le n}}
      \arrow{d}[swap]{J\vec \phi_{n+1}}
      &
      (X,\xi,x_0)
      \arrow{d}{h}
      \\
      J(\vec Q_{n+2},\vec q_{n+1})
      \arrow{r}{[y_k]_{k\le n+1}}
      &
      (Y,\zeta,y_0)
    \end{tikzcd}
  \]
  By precomposition with $\inj_n$ we obtain the following commuting diagram:
  \[
    \begin{tikzcd}[column sep=15mm]
      P_n
      \arrow{d}[swap]{\phi_n}{\cong}
      \arrow{r}{x_n}
      &
      X
      \arrow{d}{h}
      \arrow{r}{\xi}
      & TFX
      \arrow{dd}{TFh}
      \\
      Q_n
      \arrow{r}{y_n}
      \arrow{d}[swap]{q_n}
      \descto{dr}{$\sqsubseteq$}
      & Y
      \arrow{dr}{\zeta}
      \\
      FQ_{n+1}+1
      \arrow{r}[swap]{Fy_{n+1}+1}
      & FY+1
      \arrow{r}[swap]{[\eta,\bot]}
      & TFY
    \end{tikzcd}
  \]
  Applying \eqref{axChoice} to the outer part of the diagram yields
  some $f \colon P_n\to FX+1$ and the commuting diagram:
  \[
    \begin{tikzcd}[column sep=14mm,row sep=0mm]
      & X
      \arrow[bend left=10]{dr}{\xi}
      \\
      |[yshift=4mm]|
      P_n
      \arrow[bend left=10]{ur}{ x_n}
      \arrow[dashed]{r}{f}
      \arrow{d}[swap]{q_n\cdot \phi_n}
      \descto[sloped,rotate=90]{rr}{$\sqsubseteq$}
      & |[yshift=-4mm]| FX + 1
      \arrow{r}[above]{[\eta_{FX},\bot]}
      \arrow{d}{Fh+1}
      & TFX
      \arrow{d}{TFh}
      \\[5mm]
      FQ_{n+1}+1
      \arrow{r}{Fy_{n+1}+1}
      & FY+1
      \arrow{r}{[\eta_{FY},\bot]}
      & TFY
    \end{tikzcd}
    \]
   Since $q_n$ is $F+1$-precise, we obtain the following:
    \[
    \begin{tikzcd}
      P_n
      \arrow{d}[swap]{q_n\cdot \phi_n}
      \arrow{r}{f}
      & FX+1
      \descto{d}{\&}
      & X
      \arrow{d}{h}
      \\
      FQ_{n+1} +1
      \arrow{ur}[sloped,below]{Fd_{n+1}+1}
      & Q_{n+1}
      \arrow[dashed]{ur}[sloped,above]{\exists d_{n+1}}
      \arrow{r}{y_{n+1}}
      & Y
    \end{tikzcd}
  \]
  Define $d_k\colon Q_k\to X$ by $x_k\cdot \phi_k^{-1}$ for $k\le n$.
  It remains to show that
  \[
    [d_k]_{k\le n+1}\colon J(\vec Q_{n+1},\vec Q_{n})\to (X,\xi,x_0)
  \]
  is indeed a morphism in $\LCoalg(I,TF)$ and that it makes
  the desired diagram in $\LCoalg(I,TF)$ commute.
  \begin{enumerate}
  \item That $[d_k]_{k\le n+1}$ is a weak homomorphism follows from
    \autoref{lem:stepwiseFromPath} because we have:
      \[
        \begin{tikzcd}
          Q_n
          \arrow[shiftarr={yshift=6mm}]{rr}{d_{n}}
          \arrow{d}[swap]{q_n}
          \arrow{r}{\phi_n^{-1}}
          \descto{dr}{Def.~$d_{n+1}$}
          & P_n
          \arrow{r}{x_n}
          \arrow{d}{f}
          \descto{dr}{$\sqsubseteq$}
          & X
          \arrow{d}{\xi}
          \\
          FQ_{n+1} + 1
          \arrow{r}{Fd_{n+1}+1}
          & FX+1
          \arrow{r}{[\eta,\bot]}
          & TFX
        \end{tikzcd}
      \]
      and for all $k < n$:
      \[
  \tikzset{external/export next=false}
        \begin{tikzcd}
          Q_k
          \arrow[shiftarr={yshift=6mm}]{rrr}{d_{k}}
          \arrow{d}[swap]{q_n}
          \arrow{r}{\phi_k^{-1}}
          &[3mm] P_k
          \arrow{rr}{x_k}
          \arrow{d}{f}
          \descto{drr}{$\sqsubseteq$ by \autoref{lem:stepwiseFromPath}}
          & & X
          \arrow{d}{\xi}
          \\
          FQ_{k+1} + 1
          \arrow{r}{F\phi_{n+1}^{-1}+1}
          \arrow[shiftarr={yshift=-6mm}]{rr}[swap]{Fd_{k+1}}
          & FP_{k+1}+1
          \arrow{r}{Fx_{k+1}}
          & FX+1
          \arrow{r}{[\eta,\bot]}
          & TFX
        \end{tikzcd}
      \]

    \item The first desired commutativity $[d_k]_{k\le n+1}\cdot J\mathord{\vec \phi_{n+1}} =
      [x_k]_{k\le n}$ is clear by the definition of $d_k$. The second commutativity $[y_k]_{k\le
        n+1} = h\cdot [d_k]_{k\le n+1}$ is proven because the coproduct
      injections are jointly epic. For $k \le n$ we clearly have
      $y_k = h\cdot x_k\cdot \phi_k^{-1} = h\cdot d_k$. For $k=n+1$ we have
      $y_{n+1} = h\cdot d_{n+1}$ by the definition of $d_{n+1}$.

  \end{enumerate}

\subsection*{Proof of \autoref{thm:openIsHom}}
Having $\zeta\cdot h \sqsupseteq TFh\cdot \xi$ already, we only need to show:
  \[
    \begin{tikzcd}
      X
      \arrow{r}{\xi}
      \arrow{d}[swap]{h}
      & TFX
      \arrow{d}{TFh}
      \\
      Y
      \descto[sloped]{ur}{$\sqsubseteq$}
      \arrow{r}[swap]{\zeta}
      & TFY
    \end{tikzcd}
  \]
  By path-reachability, it suffices to show that for all runs
  $[x_k]_{k\le n}\colon J(\vec P_{n+1},\vec p_n)\to (X,\xi,x_0)$ we have:
  \[
    \begin{tikzcd}
      \coprod_{k\le n}P_k
      \arrow{r}{[x_k]_{k\le n}}
      \arrow{d}[swap]{[x_k]_{k\le n}}
      &
      X
      \arrow{r}{\xi}
      & TFX
      \arrow{d}{TFh}
      \\
      X
      \arrow{r}[swap]{h}
      \descto[sloped]{urr}{$\sqsubseteq$}
      &
      Y
      \arrow{r}[swap]{\zeta}
      & TFY
    \end{tikzcd}
  \]
  By \eqref{axJointlyepic}, coproduct injections are jointly
  epic in an ordered-enriched sense, so by induction, it suffices to prove that:
  \[
    \begin{tikzcd}
      P_n
      \arrow{r}{x_n}
      \arrow{d}[swap]{x_n}
      &
      X
      \arrow{r}{\xi}
      & TFX
      \arrow{d}{TFh}
      \\
      X
      \arrow{r}[swap]{h}
      \descto[sloped]{urr}{$\sqsubseteq$}
      &
      Y
      \arrow{r}[swap]{\zeta}
      & TFY
    \end{tikzcd}
  \]
  By \eqref{axJoin}, we can prove this
  by showing that for all $p_n'\colon P_n\to TFY$ we have.
  We will prove the implication
  \[
    \begin{tikzcd}
      P_n
      \arrow{r}{x_n}
      \arrow[bend right=10]{dr}[swap]{p_n'}
      &
      X
      \arrow{r}{h}
      & Y
      \arrow{d}{\zeta}
      \\
      &
      |[xshift=-10mm]|
      FY+1
      \arrow{r}{[\eta_{FY},\bot_{FY}]}
      \descto[sloped]{ur}{$\sqsubseteq$}
      & TFY
    \end{tikzcd}
    \Rightarrow
    \begin{tikzcd}
      P_n
      \arrow{r}{x_n}
      \arrow[bend right=10]{dr}[swap]{p_n'}
      &
      X
      \arrow{r}{\xi}
      & TFX
      \arrow{d}{TFh}
      \\
      &
      |[xshift=-10mm]|
      FY+1
      \arrow{r}{[\eta_{FY},\bot_{FY}]}
      \descto[yshift=2mm,sloped]{ur}{$\sqsubseteq$}
      & TFY
    \end{tikzcd}
  \]
  In the following, this implication is proved for all $p_n'$. For such a
  $p_n'\colon P_n \to FY+1$, since $P_n\in \S$, 
  \eqref{axFactor} yields an object $P_{n+1}\in \S$
  and morphisms $p_n\colon P_n\to FP_{n+1}+1$ and $y_{n+1}\colon P_{n+1}\to Y$
  with $(Fy_{n+1}+1)\cdot p_n = p_n'$. So we have a new path $(\vec P_{n+2},\vec
  p_{n+1})$ and
  a morphism defined by
  \[
    [y_k]_{k\le n+1}\colon \coprod_{k\le n+1} P_k\to Y
    \qquad
    \text{with }
    y_{k} = x_k \text{ for }k \le n,
    \text{ and }
    y_{n+1}\text{ as above}.
  \]
  This is indeed a lax homomorphism $J(\vec P_{n+2},\vec p_{n+1})\to (Y,\zeta,y_0)$
  because of \autoref{lem:stepwiseFromPath} and because for all $k<n$:
  \[
    \begin{tikzcd}
      |[xshift=-6mm]|
      P_k
      \arrow[shiftarr={yshift=6mm}]{rr}{y_k}
      \arrow{r}{x_k}
      \arrow[rounded corners,to path={
        -- ([xshift=-4mm]\tikztotarget.west |- \tikztostart)
        -- ([xshift=-4mm]\tikztotarget.west) \tikztonodes
        -- (\tikztotarget)
      }]{dd}[swap]{p_k}
      & X
      \arrow{r}{h}
      \arrow{d}[swap]{\xi}
      & Y
      \arrow{d}{\zeta}
      \\
      FX+1
      \arrow{r}{[\eta,\bot]}
      \arrow{dr}[sloped,above]{Fh+1}
      \descto[sloped]{ur}{$\sqsubseteq$}
      & TFX
      \arrow{r}{TFh}
      \descto[sloped]{ur}{$\sqsubseteq$}
      & TFY
      \\
      FP_{k+1} +1
      \arrow{u}{Fx_{k+1}+1}
      \arrow{r}[sloped,below]{Fy_{k+1}+1}
      &  FY+1
      \arrow{ur}[swap]{[\eta,\bot]}
      \descto{u}{Nat.}
    \end{tikzcd}
  \]
  and
  \[
    \begin{tikzcd}
      P_n
      \arrow{r}{x_n}
      \arrow{dr}{p_n'}
      \arrow{d}[swap]{p_n}
      &
      X
      \arrow{r}{h}
      & Y
      \arrow{d}{\zeta}
      \\
      FP_{n+1}+1
      \arrow{r}[swap]{Fy_{n+1}+1}
      &
      |[xshift=-10mm]|
      FY+1
      \arrow{r}[swap]{[\eta_{FY},\bot_{FY}]}
      \descto[sloped]{ur}{$\sqsubseteq$}
      & TFY
    \end{tikzcd}
  \]
  Since $y_k = x_k$, for $k\le n$,
  the constructed lax homomorphism makes the following
  square commute
  \[
    \begin{tikzcd}
      J(\vec P_{n+1},\vec p_n)
      \arrow{r}{[x_k]_{k\le n}}
      \arrow{d}[swap]{J(\id_{P_k})_{k\le {n+1}}}
      &[8mm] (X,\xi,x_0)
      \arrow{d}{h}
      \\
      J(\vec P_{n+2},\vec p_{n+1})
      \arrow{r}{[y_k]_{k\le n+1}}
      & (Y, \zeta)
    \end{tikzcd}
    \text{ in }\LCoalg(I,TF).
  \]
  Since $h$ is open, we obtain a diagonal lifting, that is, a lax homomorphism
  $[d_k]_{k\le n+1}\colon J\vec P_{n+1}\to (X,\xi,x_0)$ with $h\cdot d_k = y_k$
  for $k\le n+1$ and $d_k = x_k$ for $k\le n$.
  This finally proves the desired implication:
  \[
    \tikzset{external/export next=false}
    \begin{tikzcd}[column sep=12mm, baseline=(TFY.base)]
      P_n
      \arrow{rr}{d_n = x_n}
      \arrow{d}[swap]{p_n}
      \descto{drr}{$\sqsubseteq$ by \autoref{lem:stepwiseFromPath}}
      \arrow[rounded corners,to path={
        -- ([xshift=-8mm]\tikztostart.west)
        -- ([xshift=-8mm]\tikztostart.west |- \tikztotarget) \tikztonodes
        -- (\tikztotarget)
      }]{ddr}[swap]{p_n'}
      & & X
      \arrow{d}{\xi}
      \\[4mm]
      FP_{n+1} + 1
      \arrow{r}{Fd_{n+1}+1}
      \arrow{dr}[swap]{Fy_{n+1}+1}
      & FX + 1
      \arrow{d}[swap]{Fh+1}
      \arrow{r}{[\eta,\bot]}
      \descto{dr}{Nat.}
      & TFX
      \arrow{d}{TFh}
      \\
      & FY+1
      \arrow{r}{[\eta,\bot]}
      & |[alias=TFY]| TFY
    \end{tikzcd}
  \]

\subsection*{Remarks on \autoref{ass:reachability}}
The category of sets and the category of nominal sets have
(epi,mono)-factorizations. However, the assumption is rather unusual, because in
most categories, a morphism that is both epi and mono is not necessarily an
isomorphism. By the assumption of (epi,mono)-factorizations together with the
cocompleteness of $\C$, the following are equivalent for any family $(e_i\colon
X_i\to Y)_{i\in I}$:
\begin{enumerate}
  \item $(e_i)_{i\in I}$ are jointly epic.
  \item $[e_i]_{i\in I}\colon \coprod_{i\in I}X_i \to Y$ is an epimorphism.
  \item $[e_i]_{i\in I}\colon \coprod_{i\in I}X_i \to Y$ is a strong
    epimorphism.
  \item $(e_i)_{i\in I}$ are jointly strong epic.
\end{enumerate}
All the results in \autoref{subsec:reachability} work when replacing `jointly
epic' by `jointly strong epic' and by replacing `(epi,mono)-factorizations' by
the very common assumption of `(strong epi,mono)-factorizations'. Since in all
our instances (possibly sorted Sets, Nominal Sets), all epimorphisms are strong,
we phrase the results and proofs in terms of (jointly) epic families for the
sake of simplicity.

\subsection*{Proof of \autoref{prop:JointlyEpicToNoProper}}
  Consider a subcoalgebra $h\colon (Y,\zeta,y_0) \to (X,\xi,x_0)$.
  For every run $[x_k]_{k\le n}\colon$ $J(\vec P_{n+1},\vec p_n)\to (X,\xi,x_0)$,
  we have the commuting square:
  \[
    \begin{tikzcd}
      J0 = (I,\bot_{FI}\cdot!,\id_I)
      \arrow{d}[swap]{J!}
      \arrow{r}{!}
      & (Y,\zeta,y_0)
      \arrow{d}{h}
      \\
      J(\vec P_{n+1}, \vec p_n)
      \arrow{r}{[x_k]_{k\le n}}
      & (X,\xi,x_0)
    \end{tikzcd}
  \]
  and since $h$ is open by \autoref{thm:homIsOpen}, we have a run $[d_k]_{k\le
    n}$ of $(\vec P_{n+1},\vec p_n)$ in $(Y,\zeta,y_0)$. Since $[x_k]_{k\le n}$
  is epic and $h\cdot d_k = x_k$ for all $0\le k \le n$, we have
  that the mono $h$ is an epimorphism and hence an isomorphism.
  
\subsection*{Proof of \autoref{prop:noProperReachable}}
For functors $H\colon \C\to \C$ preserving intersections, we have breadth-first-search:
\begin{lemma}
  \label{bfsCoalgebra}
  Let $H\colon\C\to \C$ preserve arbitrary intersections and let $\C$ have
  $(\E,mono)$-factorizations, countable coproducts, and arbitrary intersections.
  For an $I$-pointed $H$-coalgebra $I\overset{i}{\to} X\overset{\xi}{\to} HX$,
  there are monomorphisms $(m_k\colon X_k\rightarrowtail X)$ with the property
  that:
  \begin{enumerate}
  \item $m_0$ is the mono-part of the ($\mathcal{E}$,mono)-factorization of $i$.
  \item $m_{k+1}$ is the least subobject of $X$ such that $\xi\cdot m_k$ factors
    through $Hm_{k+1}$.
  \item The union of the $m_k$, i.e.~the image of $[m_k]_{k\ge 0}\colon
    \coprod_{k\ge 0} X_k\to X$ is the carrier of a reachable subcoalgebra of
    $(X,\xi,x_0)$.
    In particular, if $(X,\xi,x_0)$ is reachable, then $[m_k]_{k\ge 0} \in \mathcal{E}$.
  \end{enumerate}
\end{lemma}
Intuitively, $X_k\rightarrowtail X$ contains precisely the states that are $k$
steps away from the initial state $X_0$, and so the union $\bigcup_{k\ge 0}X_k$
contains all reachable states.
\begin{proof}
  \textbf{1. and 2.}
  Let $I\overset{i'}{\twoheadrightarrow} X_0\overset{m_0}\rightarrowtail X$ be
  the ($\mathcal{E}$,mono)-factorization of $\map{i}{I}{X}$. For the inductive
  step, assume $m_k\colon X_k\rightarrowtail X$. Consider the family of
  subobjects $m'\colon X'\rightarrowtail X$ with the property that there exists
  some $\xi'\colon X_k\to HX'$ with $Hm'\cdot \xi' = \xi\cdot m_k$.
  Denote the intersection of these subobjects by
  $m_{k+1}\colon X_{k+1}\rightarrowtail X$, i.e.~we have the wide pullback:
  
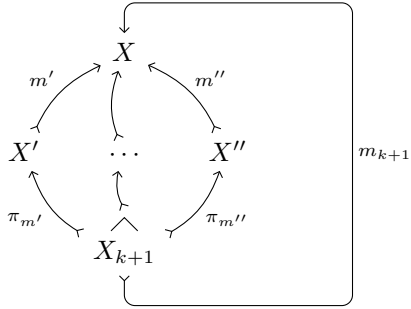
\[
    \begin{tikzcd}[column sep=4mm,row sep=8mm]
      & X
      \\
      X'
      \arrow[>->,bend left=20]{ur}{m'}
      & \cdots
      \arrow[>->,bend left=20]{u}{}
      & X''
      \arrow[>->,bend right=20]{ul}[swap]{m''}
      \\
      & X_{k+1}
      \arrow[>->,bend left=20]{ul}{\pi_{m'}}
      \arrow[>->,bend left=20,start anchor={[yshift=3mm]north}]{u}{}
      \arrow[>->,bend right=20]{ur}[swap]{\pi_{m''}}
      \pullbackangle{90}
      \arrow[>->,
      rounded corners,
      to path={
        |- ([xshift=3cm,yshift=-4mm]\tikztostart.south)
        -- ([xshift=3cm,yshift=4mm]\tikztotarget.north) \tikztonodes
        -| (\tikztotarget.north)
      }]{uu}[swap]{m_{k+1}}
    \end{tikzcd}
  \]
  Note that since $H$ preserves intersections, i.e.~pullbacks of monomorphisms,
  it preserves monomorphisms, because a morphism $m$ is mono. So the above wide
  pullback of monos is mapped again to a pullback of monos, and we have the
  commutative diagram:
  
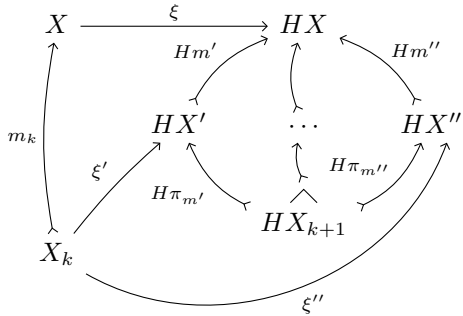
\[
    \begin{tikzcd}[column sep=4mm,row sep=8mm]
      X
      \arrow{rr}{\xi}
      &[3mm] & HX
      \\
      &HX'
      \arrow[>->,bend left=20]{ur}{Hm'}
      & \cdots
      \arrow[>->,bend left=20]{u}{}
      & HX''
      \arrow[>->,bend right=20]{ul}[swap]{Hm''}
      \\
      |[yshift=-4mm]|
      X_k
      \arrow[>->,bend left=10]{uu}{m_k}
      \arrow[bend left=5]{ur}{\xi'}
      \arrow[bend right=50,end anchor={[xshift=2mm]south}]{urrr}[swap]{\xi''}
      &
      & HX_{k+1}
      \arrow[>->,bend left=20]{ul}{H\pi_{m'}}
      \arrow[>->,bend left=20,start anchor={[yshift=3mm]north}]{u}{}
      \arrow[>->,bend right=20]{ur}{H\pi_{m''}}
      \pullbackangle{90}
    \end{tikzcd}
  \]
  By the universal property of the wide pullback, there is a unique map
  $\xi_{k}\colon X_k\to HX_{k+1}$ with $Hm_{k+1}\cdot \xi_k = \xi\cdot m_k$.
  Indeed, $m_{k+1}\colon X_{k+1}\rightarrowtail X$ is the least subobject of $X$
  with this property, because any other $m'\colon X'\rightarrowtail X$ with this
  property is itself included in the diagram for the intersection and hence
  $X_{k+1}\rightarrowtail X'$.

  \textbf{3.}
  Consider the ($\mathcal{E}$,mono)-factorization $\coprod_{k\ge 0} X_k
  \overset{e}\twoheadrightarrow Y \overset{h}\rightarrowtail X$ of $[m_k]_{k\ge
    0}$, then with the $\xi_k$ as above, we have the commutative diagram:
  \[
    \begin{tikzcd}
      &X
      \arrow{rr}{\xi}
      &[5mm]
      &[5mm] HX
      \\
      I
      \arrow[bend left]{ur}{i}
      \arrow[bend right]{dr}[swap]{\inj_0\cdot i'}
      \arrow{r}{e\cdot i'\cdot \inj_0}
      & |[xshift=-1cm]| Y
      \arrow[>->]{u}[swap]{h}
      &
      & HY
      \arrow[>->]{u}[swap]{Hh}
      \\
      & \displaystyle\coprod_{k\ge 0} X_k
      \arrow{uu}[swap]{[m_k]_{k\ge 0}}
      \arrow[->>]{u}[swap]{e}
      \arrow{r}{\coprod_{k\ge 0}\xi_k}
      & \displaystyle\coprod_{k\ge 0} HX_{k+1}
      \arrow{r}{[H\inj_{k+1}]_{k\ge 0}}
      & H\displaystyle\coprod_{k\ge 0} X_k
      \arrow[->]{u}[swap]{He}
      \arrow[shiftarr={xshift=1cm}]{uu}[swap]{H[m_k]_{k\ge 0}}
    \end{tikzcd}
  \]
  and since $e\in \mathcal{E}$ is left-orthogonal to the mono $Hh$, there is a
  unique morphism $y \colon Y\to HY$ making $e$ and $h$ coalgebra homomorphisms.
  It remains to show that $(Y,y,e\cdot \inj_0\cdot i')$ does not have a proper subcoalgebra.
  Assume a pointed subcoalgebra $f\colon (Z,\zeta,z_0) \to (Y,y,e\cdot i')$,
  with $f\colon Z\rightarrowtail Y$. We show that $e$ factors through $f$, by
  constructing maps $d_k\colon X_k\to Z$ with $e\cdot \inj_k = f\cdot d_k$ inductively.
  For $k=0$, we have the following commuting square
  which uniquely induces $d_0$:
  \[
  \begin{tikzcd}
    I
    \arrow[->>]{r}{i'}
    \arrow{d}[swap]{z_0}
    & X_0
    \arrow[dotted]{dl}[description]{\exists! d_0}
    \arrow{d}{e\cdot \inj_0}
    \\
    Z
    \arrow[>->]{r}{f}
    & Y
  \end{tikzcd}
  \]
  Given $d_k\colon X_0\to Z$ with $f\cdot d_k = e\cdot \inj_k$, we have the
  commutative diagram:
  \[
    \begin{tikzcd}
      &
      |[xshift=-2cm]|
      X
      \arrow{r}{\xi}
      & HX
      \\
      X_k
      \arrow[>->]{ur}{m_k}
      \arrow{r}{\xi_k}
      \arrow{dd}[swap]{d_k}
      \arrow{dr}{e\cdot \inj_k}
      &
      HX_{k+1}
      \arrow{dr}[swap]{H(e\cdot \inj_{k+1})}
      \arrow[>->]{ur}{Hm_{k+1}}
      \\
      &
      |[xshift=-2cm]|
      Y \arrow{r}[pos=0.3]{y}
      & HY
      \arrow[>->]{uu}[swap]{Hh}
      \\
      Z
      \arrow[>->]{ur}[swap]{f}
      \arrow{r}{\zeta}
      & HZ
      \arrow[>->]{ur}{Hf}
    \end{tikzcd}
  \]
  So $h\cdot f$ is a subobject with the property that $\xi\cdot m_k$ factors
  through $H(h\cdot f)$. Since by item (2.), $X_{k+1}$ is the least subobject
  with this property, there is some $d_{k+1}\colon X_{k+1}\rightarrowtail Z$
  with $(h\cdot f)\cdot d_{k+1} = m_{k+1}$ and in particular $f\cdot d_{k+1} =
  e\cdot \inj_{k+1}$ since $h$ is monic.

  In total, we have $f\cdot [d_k]_{k\ge 0} = [e\cdot \inj_{k}]_{k\ge 0} = e \in
  \mathcal{E}$, and so $f\in \mathcal{E}$ (by $\E$-laws in factorizations) and
  hence the mono $f$ must be an isomorphism. \qed
\end{proof}
We now can continue with the proof of the main statement, and in the following
we instantiate \autoref{bfsCoalgebra} with $H = TF$, $\mathcal{E} = \text{epi}$.

\begin{proof}[of \autoref{prop:noProperReachable}]
  Let $\xi_k\colon X_k\to TFX_{k+1}$ the witness of item 2 of
  \autoref{bfsCoalgebra}, and so $TFm_{k+1}\cdot \xi_k = \xi\cdot
  m_k$ and $m_{k+1}$ is the least subobject with this property. In the
  following, we choose sets $\mathcal{F}_k$, $k\ge 0$, containing morphisms
  $(E,e)\in \mathcal{F}_k$, $e\colon E\to X_k$, $E\in \S$, inductively:
  \begin{enumerate}
    \item $\mathcal{F}_0$ contains only $x_0'\colon I\twoheadrightarrow X_0$,
      where $x_0 = m_0\cdot x_0'$ is the pointing $I\to X$, as provided by item
      1 of \autoref{bfsCoalgebra}.
    \item For every $(E,e) \in \mathcal{F}_k$, let $G_{(E,e)}\subseteq
      \C(E,FX_{k+1}+1)$ be the set of morphisms $g\colon E\to FX_{k+1}+1$ with:
      \[
        \begin{tikzcd}
          E
          \arrow{r}{e}
          \arrow{dr}[swap]{g}
          & X_k \arrow{r}{\xi_k}
          & TFX_{k+1}
          \\
          & FX_{k+1}+1
          \descto{u}{\upinclusion}
          \arrow{ur}[swap]{[\eta,\bot]_{FX_{k+1}}}
        \end{tikzcd}
      \]
      For each $g\in G_{(E,e)}$, choose some object $Y_g\in \S$, and morphisms $y_g\colon
      Y_g\to X_{k+1}$ and a $F+1$-precise $p_g\colon E\to FY_g+1$ with $g = (Fy_g+1)\cdot p_g$,
      according to \eqref{axFactor}.
      Define $\mathcal{F}_{k+1} := \{ (Y_g,y_g)\mid (E,e) \in \mathcal{F}_k, g\in G_{(E,e)}\}$.
\end{enumerate}
We need to prove some properties about the $\mathcal{F}_k$
\begin{itemize}
\item By construction, for every $(E',e') \in \mathcal{F}_{k+1}$, there is some
  $(E,e) \in \mathcal{F}_k$ and an $F+1$-precise map $p_g\colon E\to FE'+1$ with
  \[
    \begin{tikzcd}[row sep=9mm]
      E
      \arrow{dr}{g}
      \arrow{dd}[swap]{p_g} \arrow{rr}{e} &
      \descto[xshift=4mm]{ddr}{$\sqsubseteq$}
      & X_k \arrow{dd}{\xi_k}
      \\
      & FX_{k+1}+1
      \arrow{dr}[swap]{[\eta,\bot]}
      \\
      FE'+1 \arrow[swap]{r}{[\eta,\bot]}
      \arrow{ur}{Fe'+1}
      & TFE'
      \arrow[swap]{r}{TFe'} & TFX_{k+1}
    \end{tikzcd}
  \]
  Consequently by \autoref{lem:stepwiseFromPath}, for every $(E,e) \in
  \mathcal{F}_k$, $k\ge 0$, there is a run $[x_j]_{j\le k}\colon J(\vec
  P_{k+1}, \vec p_k) \to (X,\xi,x_0)$ with $P_k = E$, $x_k = \inj_k\cdot e$.

\item For every family $\mathcal{F}_k$, $k \ge 0$, the morphism $[e]_{(E,e)\in
    \mathcal{F}_k}$ is an isomorphism; this means that for the factorization:
      \[
        [e]_{(E,e)\in \mathcal{F}_{k}}
        ~\equiv~
        \big(
        \coprod_{\mathclap{(E,e)\in \mathcal{F}_{k}}} E
        \overset{q_k}\twoheadrightarrow
        \bar X_{k}
        \overset{s_k}\rightarrowtail X_{k}
        \big)
      \]
      we have that $s_k$ is an isomorphism for every $k\ge 0$.
      \allowdisplaybreaks
      \begin{align*}
        \xi_k\cdot e~~
        &\overset{\mathclap{\text{\eqref{axJoin}}}}=
          \quad\bigsqcup_{g\in G_{(E,e)}} [\eta_{FX_{k+1}},\bot_{FX_{k+1}}]\cdot g
           \\ &
          \overset{\mathclap{\text{Def. }p_g,y_g}}=
          \quad\bigsqcup_{g\in G_{(E,e)}} [\eta_{FX_{k+1}},\bot_{FX_{k+1}}]\cdot (Fy_g+1)\cdot p_g
           \\ & \overset{\mathclap{(Y_G,y_g) \in \mathcal{F}_{k+1}}}
         = \quad\bigsqcup_{g\in G_{(E,e)}} [\eta_{FX_{k+1}},\bot_{FX_{k+1}}]\cdot (F(s_{k+1}\cdot q_{k+1}\cdot \inj_{y_g})+1)\cdot p_g
                &
           \\ & \overset{\mathclap{\text{Naturality}}}
         = \quad\bigsqcup_{g\in G_{(E,e)}} TFs_{k+1}\cdot [\eta_{FX_{k+1}'},\bot_{FX_{k+1}'}]\cdot (F(q_{k+1}\cdot \inj_{y_g})+1)\cdot p_g
           \\ &\overset{\mathclap{\bigsqcup-\text{reflection}}}
         = \qquad TFs\cdot \underbrace{\bigsqcup_{g\in G_{(E,e)}} [\eta_{FX_{k+1}'},\bot_{FX_{k+1}'}]\cdot (F(q_{k+1}\cdot \inj_{y_g})+1)\cdot p_g}_{\text{short hand }v_{(E,e)}\colon E\to TF\bar X_{k+1}}
      \end{align*}
      Hence, $TFs\cdot v_{(E,e)} = \xi_k\cdot e$ for all $(E,e)\in
      \mathcal{F}_k$. Since the $e\in \mathcal{F}_k$ are jointly epic by the
      induction hypothesis, we have a unique diagonal $\xi_k'$ in the following
      square:
      \[
        \begin{tikzcd}[column sep=12mm]
          \coprod_{(E,e)\in \mathcal{F}_k} E
          \arrow[->>]{r}{[e]_{(E,e)\in \mathcal{F}_k}}
          \arrow{d}[swap]{[v_{(E,e)}]_{(E,e)\in \mathcal{F}_k}}
          & X_k
          \arrow{d}{\xi_k}
          \arrow[dashed]{dl}[description]{\exists !\xi_k'}
          \\
          TFX_{k+1}
          \arrow[>->]{r}{TFs_{k+1}}
          & TFX_{k+1}
        \end{tikzcd}
      \]
      Since $\xi_k$ was constructed to be the least morphism, $s_{k+1}$ is necessarily
      an isomorphism.
\end{itemize}
Now we have that $f_k := [e]_{(E,e)\in \mathcal{F}_k}\colon \coprod_{(E,e)\in
  \mathcal{F}_k} E\to X_k$ is epic for every $k\ge 0$. Since $(X,\xi,x_0)$ is
reachable, $[m_k]_{k\ge 0}\colon \coprod_{k\ge 0} X_k\to X$ (cf.~item 3 of \autoref{bfsCoalgebra})
is an epimorphism.
Hence, the familiy $(m_k\cdot f_k)_{k\ge 0 }$ is jointly epic; this family is
contained in the family of runs in $(X,\xi,x_0)$, and so the family of runs is
jointly epic.
  \qed
\end{proof}

\subsection*{Proof of \autoref{prop:analyticPrecise}}
  
  The first step is to describe when a map is $\B$-precise:
  \begin{lemma} \label{B-precise}
  A map $s\colon X\to \Bag Y$ is $\Bag$-precise iff
  for all $y\in Y$, $$\sum_{x\in S} s(x)(y) = 1.$$
\end{lemma}
\begin{proof}
  \emph{Sufficiency} is proven by the two inequalities:
  \begin{description}
    \item[$(\le 1)$]
      Let $t_{X,Y}\colon \Bag Y\times X\to \Bag(Y\times X)$ be the functorial strength
      defined as for all $\Set$-functors by $t_{X,Y}(b,x) = \Bag(y\mapsto (y,x))(b)$.
      We thus have the commuting diagram
      \[
        \begin{tikzcd}
          X \arrow{r}{\fpair{s,\id_X}}
          \arrow{drr}[swap]{s}
          & \Bag Y\times X
          \arrow{r}{t_{X,Y}}
          & \Bag (Y\times X)
          \arrow{d}{\Bag \pi_1}
          \\
          & & \Bag Y
        \end{tikzcd}
      \]
      Since $s$ is $\Bag$-precise, we obtain a map $d\colon Y\to Y\times X$ with
      $\pi_1\cdot d = \id_Y$ and $t_{X,Y}\cdot \fpair{s,\id_X} = \Bag d\cdot s$,
      that means $\Bag(y\mapsto (y,x))(s(x)) = \Bag(y\mapsto (y,d'(y)))(s(x))$,
      with $d' = \pi_2\cdot d$.
      Hence, every $y\in Y$ appears in at most one $s(x)$, namely only in
      $s(d'(y))$ or not at all; this means for every $y\in Y$, only one summand
      of $\sum_{x\in S}s(x)(y)$ is non-zero.
      
      To see that every summand of $\sum_{x\in S}s(x)(y)$ is at most 1, define
      $f\colon X\to \Bag(\N\times Y)$ by
      \[
        f(x)(n,y) = \begin{cases}
          1 & \text{if }n < s(x)(y) \\
          0 & \text{otherwise}
        \end{cases}
      \]
      and so 
      \[
        \Bag(\pi_2)(f(x))(y) = \sum_{\mathclap{(n,y) \in \N\times Y}} f(x)(n,y)
        = \sum_{\mathclap{\substack{(n,y) \in \N\times Y\\ n < s(x)(y)}}} 1
        = s(x)(y)
      \]
      Hence, we have some map $d\colon Y\to \N\times Y$ with $f = \Bag(d) \cdot
      s$ and $d(y) = (d'(y), y)$ for some $d'\colon Y\to \N$. The first equality
      expands to
      \[
        f(x)(n,y) = \Bag(d)(s(x))(n,y) = \sum_{\mathclap{\substack{y'\in Y\\ d(y') = (n,y)}}} s(x)(y')
        = \begin{cases}
          s(x)(y)
          & \text{if }n = d'(y)
          \\
          0 &\text{otherwise}
        \end{cases}
      \]
      This implies that for every $x$ and $y$ there is at most one $n\in N$ such
      that $f(x)(n,y) > 0$, hence $s(x)(y) \le 1$.

    \item[$(\ge 1)$]
       Assume that $\sum_{x\in S}s(x)(y) = 0$ for some $y\in Y$.
       Define $s'\colon X\to \Bag (Y\setminus \{y\})$ by
       $s'(x)(y') = s(x)(y')$. With the obvious inclusion $i\colon
       Y\setminus \{y\} \hookrightarrow Y$ we have $s = \Bag i\cdot s'$ and so
       the $\Bag$-precise $s$ induces a map $d\colon Y\to Y\setminus\{y\}$ with
       $i\cdot d = \id_Y$, a contradiction.

  \end{description}
  For \emph{necessity}, consider $f$ and $h$ with
  \begin{equation}
    \begin{tikzcd}
      X
      \arrow{r}{f}
      \arrow{d}[swap]{s}
      & \Bag C
      \arrow{dl}{\Bag h}
      \\
      \Bag Y
    \end{tikzcd}
    \label{eq:sfh}
  \end{equation}
  By assumption on $s$ we have for all $y\in Y$:
  \[
    1 = \sum_{x\in X}s(x)(y)
    \overset{\text{\eqref{eq:sfh}}}{=}
    \sum_{x\in X}\sum_{\substack{c\in C\\ h(c)= y}} f(x)(c)
    = \sum_{\substack{x\in X,\,c\in C\\ h(c)= y}} f(x)(c)
  \]
  Since all summands of the right-hand sum is are non-negative integers, there
  exists precisely one $x_y\in X, c_y\in C$ with $h(c_y) = y$ such that $f(x)(c)
  \neq 0$. Hence, define $d\colon Y\to C$ by this witness: $d(y) = c_y$. This
  implies directly that $h(d(y)) = y$ and that $s(x)(y) = f(x)(d(y))$ for all
  $x\in X$, $y\in Y$.
  \[
    \Bag d(s(x))(c)
    = \sum_{\mathclap{y\in Y, d(y)=c}} s(x)(y)
    = \sum_{\mathclap{y\in Y, d(y) = c}} f(x)(d(y))
    = f(x)(c)
  \]
  where the last equality holds because $d$ is injective.\qed
\end{proof}
Now, we want to transfer preciseness from the bag functor to any analytic functor 
using the natural transformation $\alpha$. This is done using the following:
\begin{lemma} \label{preciseAlognCartesian}
  Let $F$ preserve weak pullbacks and let $\alpha\colon F\to G$ be a natural
  transformation whose naturality squares are weak pullbacks. Then a
  morphism $f\colon X\to FY$ is $F$-precise iff $\alpha_Y\cdot f$ is
  $G$-precise. If $G$ admits precise factorizations w.r.t.~$\S\subseteq \obj
  \C$, then so does $F$.
\end{lemma}
\begin{proof}
  For \textbf{sufficiency}, let $f\colon X\to FY$ be $F$-precise and consider a commuting
  diagram
  \[
    \begin{tikzcd}
      X
      \arrow{r}{g}
      \arrow{d}[swap]{\alpha_Y\cdot f}
      & GW
      \arrow{d}{Gw}
      \\
      GY
      \arrow{r}{Gy}
      & GZ
    \end{tikzcd}
  \]
  Since the naturality square for $w$ is a weak pullback, a morphism $g'\colon
  X\to FZ$ exists, making the following diagram commute:
  \[
    \begin{tikzcd}
      X
      \arrow[dashed]{dr}[swap,near end]{g'}
      \arrow{dd}[swap]{f}
      \arrow[bend left=15]{drr}{g}
      \\
      & FW
      \pullbackangle{-45}
      \arrow{r}{\alpha_W}
      \arrow{d}[swap]{Fw}
      & GW
      \arrow{d}{Gw}
      \\
      FY
      \arrow{r}{Fy}
      \arrow{dr}[swap]{\alpha_Y}
      & FZ
      \arrow{r}{\alpha_Z}
      & GY
      \\
      & GZ
      \arrow{ur}[swap]{Gy}
    \end{tikzcd}
  \]
  Since $f$ is $F$-precise, we obtain a morphism $d\colon Y\to W$ with $w\cdot
  d=y$ and $Fd\cdot f = g'$, and thus also $g = \alpha_W\cdot g' =
  \alpha_W\cdot Fd\cdot f = Gd\cdot \alpha_Y\cdot f$ as desired.
  
  For \textbf{necessity}, we use the simplified version of $F$-precise because
  $F$ preserves weak pullbacks. So let $\alpha_Y\cdot f$ be $G$-precise and
  consider a commuting diagram
  \[
    \begin{tikzcd}
      X
      \arrow{r}{g}
      \arrow{d}[swap]{f}
      & FZ
      \arrow{dl}{Fz}
      \\
      FY
    \end{tikzcd}
  \]
  Hence, $Gz\cdot \alpha_Z\cdot g = \alpha_Y\cdot f$ and since $\alpha_Y\cdot f$ is $G$-precise
  we obtain a morphism $d\colon Y\to Z$ with
  \[
    \begin{tikzcd}
      X
      \arrow{r}{g}
      \arrow{d}[swap]{f}
      & FZ
      \arrow{r}{\alpha_Z}
      & GZ
      \\
      FY
      \arrow{r}{\alpha_Y}
      & GY
      \arrow{ur}[swap]{Gd}
    \end{tikzcd}
    \&
    \quad
    z\cdot d = \id_Y
  \]
  The naturality square for $d$ is a weak pullback, and so we have some $f'\colon
  X\to FY$ making the following diagrams commute:
  \[
    \begin{tikzcd}
      X
      \arrow[bend left=15]{drr}{g}
      \arrow[dashed]{dr}[swap,near end]{\exists f'}
      \arrow{d}[swap]{f}
      \\
      FY
      \arrow[bend right=10]{dr}[swap]{\alpha_Y}
      & FY
      \arrow{r}{Fd}
      \arrow{d}[swap]{\alpha_Y}
      \pullbackangle{-45}
      & FZ
      \arrow{d}{\alpha_Z}
      \\
      & GY
      \arrow{r}{Gd}
      & GZ
    \end{tikzcd}
    \quad\Longrightarrow\quad
    \begin{tikzcd}
      & X
      \arrow{dl}[swap]{f'}
      \arrow{d}[swap]{g}
      \arrow{dr}{f}
      \\
      FY
      \arrow{r}[swap]{Fd}
      \arrow[shiftarr={yshift=-5mm}]{rr}[swap]{F\id_Y}
      & FZ
      \arrow{r}[swap]{Fz}
      & FY
    \end{tikzcd}
  \]
  So $d$ is indeed a diagonal lifting: $Fd\cdot f= Fd\cdot f' = g$.

  For \textbf{precise factorizations}, consider $f\colon X\to FY$ with $X\in
  \S$. For the $G$-precise factorization $Gy'\cdot g'$ of $\alpha_Y\cdot f$ we
  have:
  \[
    \begin{tikzcd}[baseline=(mybase.base)]
      |[yshift=8mm]|
      X
      \arrow[bend left=15]{rr}{g'}
      \arrow[bend right=15]{dr}[swap]{f}
      & FY' \arrow{r}{\alpha_{Y'}}
       \arrow{d}[swap]{Fy'}
       \pullbackangle{-45}
      & GY'
      \arrow{d}[alias=mybase]{Gy'}
      \\
      & FY \arrow{r}{\alpha_Y}
      & GY
    \end{tikzcd}
    \quad Y'\in \S
  \]
  Since the naturality square is a weak pullback, we obtain some $f'\colon X\to FY'$
  with $Fy'\cdot f' = f$ and $Y'\in \S$ as desired.
  \qed
\end{proof}

To conclude, since analytic functors preserve weak pullbacks, it is enough to prove the 
following:
\begin{lemma}
  The naturality squares of $\alpha\colon F\to \Bag$ are weak pullback squares.
\end{lemma}
\begin{proof}
  Consider a map $f\colon X\to Y$ and elements $t \in FY$, $s\in \Bag X$ with
  $\alpha_Y(t) = \Bag f(s)$. We can write $t$ as an equivalence class $t =
  [\sigma(y_1,\ldots,y_n)]$. Every $y_i$ is in the image of $f$, and so any
  $t' = [\sigma(x_1,\ldots,x_n)] \in FX$ with $f(x_i) = y_i$ for all $1\le i\le n$ has
  the property that $\alpha_X(t') = s$ and $Ff(t') = t$.\qed
\end{proof}
So we obtain \autoref{prop:analyticPrecise} as a corollary of \autoref{preciseAlognCartesian}.

\subsection*{Auxiliary Lemma in Nominal Sets}
In this section, we will use the following standard lemma on strong nominal sets
(see e.g.~\cite[Prop.~5.10]{milius16} or \cite[p.~3]{kurz10}):
\begin{lemma}
  \label{strongUMP} Given a strong nominal set $X$, a subset
  $O\subseteq X$ such that for every $x'\in X$ there is precisely one $x\in X$
  with a $\pi\in
  \perms(\A)$ fulfilling $x' = \pi\cdot x$, and a map $f\colon O\to Y$ into a nominal
  set $Y$ with $\supp(f(x))\subseteq \supp(x)\,\forall x\in O$. Then there is a
  unique equivariant map $f'\colon X\to Y$ that agrees with $f$ on $O$.
  Such a subset $O\subseteq X$ exists for every nominal set $X$.
\end{lemma}
\subsection*{Details for \autoref{ex:whyClassS}}
\subsubsection*{Factorization w.r.t.~strong nominal sets.}
Let $F\colon \Nom\to \Nom$ be the functor of unordered pairs:
\[
  FX = \{ p \subseteq X \mid 1\le |p| \le 2\}.
\]
with element-wise nominal structure. Given a strong nominal set $X$ and an
equivariant map $f\colon X\to FY$, define $s\colon X\to F(X+X)$ with
\[
  s(x) = \{\inl(x),\inr(x)\}.
\]
To see that $s$ is $F$-precise, consider a commutative square
\[
  \begin{tikzcd}
    X
    \arrow{r}{g}
    \arrow{d}[swap]{s}
    & FW
    \arrow{d}{Fu}
    \\
    F(X+X)
    \arrow{r}{F[h_1,h_2]}
    & FZ.
  \end{tikzcd}
\]
Spelling out the commutativity, we have:
\[
  \{h_1(x),h_2(x)\} = F[h_1,h_2](s(x)) = Fu(g(x))
  \text{ for all }x\in X.
\]
And so $\{h_1(x),h_2(x)\} = Fu(g(x)) = \{ u(w_1),u(w_2)\}$ for $g(x) = \{w_1,w_2\}$.
For every orbit of $X$ choose some $x\in X$, and let $d_1(x), d_2(x) \in W$ such
that $u(d_1(x)) = a_1(x)$ and $u(d_2(x)) = a_2(x)$ and $\{d_1(x),d_2(x)\} =
g(x)$ (this involves a finite choice if $|g(x)| = 2$ and $|Fu(g(x))| = 1$ but is
unique otherwise). Because
$X$ is a strong nominal set and by \autoref{strongUMP}, $d_1$ and $d_2$ extend
to equivariant maps $d_1,d_2\colon X\to W$ with $\{d_1(x),d_2(x)\} = g(x)$ and $u(d_i(x))
= h_i(x)$, $i\in \{1,2\}$. So in total $g = Fd\cdot s$ and $[h_1,h_2] = u\cdot [d_1,d_2]$.

To see that $f$ factors through $s$, consider $g:=f$, $Z=1$, $u=\mathord{!}\colon W\to Z$,
and $h_1=h_2=\mathord{!}\colon X\to 1$, using that $F1 \cong 1$, and so the above square
commutes trivially and the diagonal fill in proves that $f$ factors through $s$
as desired.

\subsubsection*{No factorization w.r.t.~all nominal sets.} In order to prove that $F$ does not admit precise factorizations w.r.t.~all
nominal sets, consider the (unique) equivariant map
\[
  f\colon P\to F1\cong 1\qquad\text{ with }P=\big\{\{a,b\}\mid a,b \in \A, a\neq
  b\big\}
\]
Let $s\colon P\to FX$ be its precise factorization, so necessarily $f = F!\cdot
s$, with $!\colon X\to 1$. In the following, we derive information about $s$ and
$X$.

As a first aspect, consider the constant equivariant map $t\colon P\to F2$ with
$t(\{a,b\}) = \{0, 1\}$. Hence,
\[
  \begin{tikzcd}
    P
    \arrow{r}{t}
    \arrow{d}[swap]{s}
    \arrow{dr}[description]{f}
    & F2
    \arrow{d}{F!}
    \\
    FX
    \arrow{r}[swap]{F!}
    & F1
  \end{tikzcd}
  \Longrightarrow
  \exists d\colon X\to 2\colon Fd\cdot s = t
\]
For every $\{a,b\}\in P$, $s(\{a,b\}) = \{x_1,x_2\}$ for some $x_1,x_2\in X$,
and $\{d(x_1),d(x_2)\} = \{0,1\}$. Since $0,1\in 2$ are in different orbits,
$x_1,x_2$ are in different orbits in $X$, i.e.~so for all $\pi \in \perms(\A)$,
$x_1\neq \pi\cdot x_2$.

As a second aspect, consider $p\colon P\to F(\A^2)$ defined by
\[
  p(\{a,b\}) = \{(a,b),(b,a)\}.
\]
Again, we can use the commutativity:
\[
  \begin{tikzcd}
    P
    \arrow{r}{p}
    \arrow{d}[swap]{s}
    \arrow{dr}[description]{f}
    & F\A^2
    \arrow{d}{F!}
    \\
    FX
    \arrow{r}[swap]{F!}
    & F1
  \end{tikzcd}
  \Longrightarrow
  \exists u\colon X\to \A^2\colon Fu\cdot s = p
\]
For $\{a,b\}\in P$ and $s(\{a,b\}) = \{x_1,x_2\}$ as before, $\{u(x_1),u(x_2)\}
= \{(a,b),(b,a)\}$ and the equivariance and functionality of $u$ implies
\[
  u(x_1)\neq u(x_2) = (a\,b)\cdot u(x_1) = u((a\,b)\cdot x_1)
  \Longrightarrow
  (a\,b)\cdot x_1 \neq x_1.
\]
Together with $\sigma\cdot x_1 \neq x_2$ from above, we have that
$(a\,b)\cdot x_1 \not\in \{x_1,x_2\}$.
This leads to a contradiction to the equivariance of $s$:
\begin{align*}
\{x_1,x_2\}
  &= s(\{a,b\}) = s((a\,b)\cdot \{a,b\})
  = (a\,b)\cdot s(\{a,b\})
    \\ &
  = (a\,b)\cdot \{x_1,x_2\}
  = \{(a\,b)\cdot x_1,(a\,b)\cdot x_2\}
  \neq \{x_1,x_2\}.
\end{align*}

\subsection*{Proof of \autoref{prop:powersetNom}}

Let us check the axioms:
\begin{itemize}
	\item \eqref{axJointlyepic}: obtain as in $\Set$.
	\item \eqref{axUnitBot}: We indeed have natural transformations 
	$\eta\colon X\to \powufs(X)$ for
    	singleton sets and $\bot\colon 1\to \powufs(X)$ for the empty set.
    \item \eqref{axJoin}: consider a equivariant $A\to \powufs B$ with $A$
      strong. Since $f$ is by definition an upper bound for the $\mathcal{F} =
      \{ [\eta_B,\bot_B]\cdot f'\sqsubseteq f \mid f'\colon A\to B+1\}$, we need
      to prove that it is the least upper bound. Let $g$ be an upper bound for
      $\mathcal{F}$. In order to have $f\sqsubseteq g$ we need to show for every
      $a\in A$ and $b\in f(a)$ that $b\in g(a)$. Since $b\in f(a)$ and $f(a)$ is
      ufs, we have $\supp(b) \subseteq \supp(f(a)) \subseteq \supp(a)$. Define
      $f'\colon A\to B+1$ via \autoref{strongUMP}, where $O \subseteq A$ is
      chosen with $a\in O$ and $f''\colon O\to B+1$ is defined by $f''(a) =
      \inl(b)$ and $f''(x) = \inr(*)$ for $x\neq a$. Now, $f'$ is the unique
      equivariant map extending $f''$. Since $[\eta_B,\bot]\cdot f''(x)
      \subseteq f(x)$ for all $x\in O$, we have that $[\eta_B,\bot]\cdot f'
      \sqsubseteq f$, and so $[\eta_B,\bot]\cdot f' \in \mathcal{F}$, and thus
      $[\eta_B,\bot]\cdot f'\sqsubseteq g$, hence $b\in g(a)$.
	\item \eqref{axChoice}: consider a weakly commuting square:
    \[
      \begin{tikzcd}
        A
        \arrow{r}{x}
        \arrow{d}[swap]{y}
        \descto[sloped,rotate=90]{dr}{$\sqsubseteq$}
        & \powufs X
        \arrow{d}{\powufs h}
        \\
        Y
        \arrow{r}{\eta_Y}
        & \powufs Y
      \end{tikzcd}
    \]
    where $A$ is strong.
    Pick any subset $O\subseteq A$ fulfilling the assumption of
    \autoref{strongUMP}. For each $a\in A$ we have that $y(a) \in h[x(a)]$, so
    for each $a \in O$, there exists some $a' \in x(a)$ with $h(a') = y(a)$.
    For each $a\in O$, denote this witness by $x'(a) := a'$.
    By \autoref{strongUMP}, $x'$ extendes to an equivariant map $c\colon A\to
    X$ with $c(a) = x'(a) \in x(a)$ for all $a \in O$.
    For every $b\in A$, there is some $\pi \in \perms(\A)$ and $a\in A$ with $b
    = \pi\cdot a$, and hence $x'(b) = x'(\pi\cdot a) = \pi\cdot x'(a) \in
    \pi\cdot x(a) = x(b)$, i.e.~$\eta_X\cdot x' \sqsubseteq x$. By construction
    of $c$, we have that $h\cdot x' = y$.
\end{itemize}

\subsection*{Proof of \autoref{prop:nomPrecise}}

To prove this, by \autoref{propFactorizationClosure}.4, it is enough to prove
that the binding functor $[\A]$ admits precise factorizations w.r.t.~strong nominal sets.

  Recall from \cite{pitts13} that $[\A]$ has the left adjoint $\A\#\_$, sending $X$
  to the nominal set $\A\#X := \{(a,x) \in \A\times X\mid a \not\in\supp(x)\}$. If $X$ is
  a strong nominal set, then so is $\A\# X$, hence $[\A]$ admits precise
  factorizations w.r.t.~strong nominal sets by \autoref{propFactorizationClosure}.5.\qed

\subsection*{Proof of \autoref{prop:LasotaAxioms}}
All the axioms are proved as in $\Set$, component-wise. The only interesting axiom is to 
prove that $F\{X_P\}_{P\in \bbP}
  = \{
  \coprod_{Q\in \bbP}
  {\bbP(P,Q)}\times X_Q\}_{P\in \bbP}$ 
  admits precise factorisations. Since $F = \coprod_{Q\in \bbP} H_Q\times\pi_Q$ with:\\
  \indent - $H_Q$ being the constant functor on $(\bbP(P,Q))_{P\in \bbP}$,\\
  \indent - $\pi_Q$ being the functor mapping $\{X_P\}_{P\in\bbP}$ to $\{X_Q\}_{P\in \bbP}$,\\
  by \autoref{propFactorizationClosure}.(1,3,4), it is enough to prove that $\pi_Q$ admits
  precise factorisations.
  
  For that purpose, given a morphism, i.e.~family of maps, $(f_{P})_{P\in \bbP}\colon \{X_P\}_{P\in \bbP} \rightarrow
  \pi_Q(\{Y_P\}_{P\in \bbP}) = \{Y_Q\}_{P\in \bbP}$, define
  \[
    (Z_P)_{P\in \bbP}
    \qquad
    \text{ with }
    Z_Q = \coprod_{P\in \bbP} X_P
    \text{ and }
    Z_P = \emptyset \text{ for }P\neq Q.
  \]
  and
  $(\inj_{P})_{P\in \bbP}\colon \{X_P\}_{P\in \bbP} \to \pi_Q(Z_P)_{P\in \bbP}$
  \[
    \inj_P\colon X_P\to (\pi_Q(Z_P)_{P\in \bbP})_P = Z_Q = \coprod_{P\in \bbP} X_P
  \]
  and $ (h_P)_{P\in \bbP}\colon (Z_P)_{P\in\bbP} \to (Y_P)_{P\in\bbP}$ with
  \[
    \quad
    h_Q = [f_P]_{P\in \bbP}\colon \coprod_{P\in \bbP} X_P\to Y_Q,
    \quad
    h_P = \mathord{!}\colon \emptyset\to Y_Q
    \text{ for }P\neq Q.
  \]
  Obviously, $(\pi_Qh)_P\cdot \inj_P = f_P$ for all $P\in \bbP$. It remains to
  show that $(\inj_P)_{P\in \bbP}$ is $\pi_Q$-precise. Consider a commutative
  square:
  \[
    \begin{tikzcd}[column sep=15mm]
      (X_P)_{P\in \bbP}
      \arrow{d}[swap]{(\inj_P)_{P\in \bbP}}
      \arrow{r}{(v_P)_{P\in \bbP}}
      & \pi_Q((V)_{P\in \bbP})
      \arrow{d}{\pi_Q((w_P)_{P\in \bbP})=w_Q}
      \\
      \pi_Q((Z)_{P\in \bbP})
      \arrow{r}{\pi_Q((z_P)_{P\in \bbP})}[below]{=\,z_Q}
      & \pi_Q((W)_{P\in \bbP})
    \end{tikzcd}
    \tag{$*$}
  \]
  Define a diagonal $(d_P)_{P\in \bbP}\colon (Z_P)_{P\in \bbP}\to (V_P)_{P\in
    \bbP}$ by
  \[
    d_Q = [v_P]_{P\in \bbP}\colon
    \coprod_{P\in \bbP} X_P\to V_Q
    \quad
    d_P = \mathord{!}\colon Z_P\to V_P\text{ for }P\neq Q.
  \]
  Obviously, we have the first triangle $v_P = d_Q\cdot \inj_P$ for all $P\in
  \bbP$. The other triangle follows by case distinction: for $P\neq Q$, $z_P =
  w_P\cdot d_P$ holds by initiallity of $Z_P = \emptyset$; we have
  \[
    z_Q\cdot \inj_P \overset{(*)}{=} w_Q\cdot v_P = w_Q\cdot d_Q\cdot \inj_P
    \quad\forall P\in \bbP
  \]
  and since the $\inj_P\colon X_P\to Z_Q$ are jointly epic, $z_Q = w_Q\cdot d_Q$.
  So $(d_P)_{P\in \bbP}$ is a diagonal lifting and so $(\inj_P)_{P\in \bbP}$ is $\pi_Q$-precise.
  
\subsection*{Proof of \autoref{prop:LasotaPrecise}}
To prove this proposition, it is enough to prove the following:
\begin{lemma}
  A morphism $f\colon \chi^P\to FY$ is $F$-precise iff $Y=\chi^Q$ for some
  $Q\in \bbP$.
\end{lemma}
\begin{proof}
  \textbf{Necessity.} Consider $f_P\colon 1\to \coprod_{Q\in \bbP}
  \bbP(P,Q)\times Y_Q$ and let $f_P = \inj_Q(m,x)$. Define the injective
  $h\colon \chi^Q\to (Y_R)_{R\in \bbP}$ by $h_R=\mathbin{!}\colon \emptyset \to
  Y_R$ for $R\neq Q$ and $h_Q= x\colon 1\to X_Q$. Furthermore, define by
  $g\colon \chi^P\to F\chi^Q$ $g_P=\inj_Q(m,*)$, and so $Fh\cdot g = f$.
  Since $f$ is $F$-precise, there is a $d\colon Y\to \chi^Q$ with $h\cdot d =
  \id_Y$, so in total $h$ is an isomorphism.

  \textbf{Sufficiency.} For $f\colon \chi^P\to F\chi^Q$ let $m\colon P\to Q$
  with $f_P = \inj_Q(m,*)$ and consider a commuting diagram:
  \[
    \begin{tikzcd}
      \chi^P
      \arrow{d}[swap]{f}
      \arrow{r}{w}
      & FW
      \arrow{d}{Fh}
      \\
      F\chi^Q
      \arrow{r}{Fz}
      & FZ
    \end{tikzcd}
    \text{ in }\Set^{\obj\bbP}
  \]
  And so in component $P$ we have:
  \[
    \begin{tikzcd}
      1
      \arrow{d}[swap]{\fpair{m,\id}}
      \arrow{rr}{w_P}
      \arrow{dr}[near end]{f_P}
      &
      & \coprod_{R\in \bbP}\bbP(P,R)\times W_R
      \arrow{d}{Fh}
      \\
      \bbP(P,Q)\times \chi^Q_Q
      \arrow{r}{\inj_Q}
      &
      \coprod_{R\in \bbP}\bbP(P,R)\times \chi^Q_R
      \arrow{r}{Fz}
      &
      \coprod_{R\in \bbP}\bbP(P,R)\times Z_R
    \end{tikzcd}
    \text{ in }\Set
  \]
  Hence, (formally by the extensivity of $\Set$) $w_P$ is necessarily $w_P =
  \inj_Q(m, x)$ for some $x\in W_Q$ with $h_Q(x) = z_Q(*)$. 
  Define $d\colon \chi^Q\to W$ by $d_Q\colon
  1 \to W_Q, d_Q = x$, and $d_R =\mathbin{!}\colon \emptyset\to W_R$ for $R\neq
  Q$. Obviously $d$ fulfils both $Fd\cdot f = w$ and $h\cdot d = z$.\qed
\end{proof}

\subsection*{Notes on \autoref{sec:furtherWork}}

In probabilistic systems, it is not clear how paths should look like.
Consider the probabilistic system $A$:
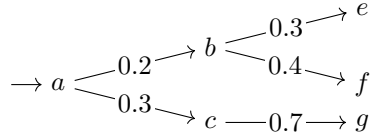
\begin{center}
\begin{tikzpicture}
  \node (a) {$a$};
  \node (b) at +(2, 0.5) {$b$};
  \node (c) at +(2,-0.5) {$c$};
  \node (e) at ($ (b) + (2, 0.5) $) {$e$};
  \node (f) at ($ (b) + (2, -0.5) $) {$f$};
  \node (g) at ($ (c) + (2, 0) $) {$g$};
  \path[->,every node/.append style={fill=white,inner sep=1pt}]
    ([xshift=-4mm]a.west) edge (a)
    (a) edge node {$0.2$} (b)
    (a) edge node {$0.3$} (c)
    (b) edge node {$0.3$} (e)
    (b) edge node {$0.4$} (f)
    (c) edge node {$0.7$} (g)
    ;
\end{tikzpicture}
\end{center}
There are different possibilities how to model a path in such a system:
\begin{enumerate}
\item Simply $\bullet \to \bullet \to \bullet$ as in LTS for a singleton input alphabet. Of
  course, this makes the probabilities in the system entirely meaningless for runs.
\item A path and a probability $(p,\bullet \to \bullet\to\bullet)$ meaning that
  there is a path whose weights multiply to at least $p$. So the above system
  has a run for such a path if $p \le 0.08 = 0.2\cdot 0.4$.

\item A path is a probabilistic system in which each state has at most one
  successor, and so the notion of run of path coincides with a functional
  simulation from that path to the system. E.g.~the above system has a run for $\bullet
  \xrightarrow{0.1}\bullet \xrightarrow{0.25} \bullet$. Clearly, there is a
  functional bisimulation $f: A\to X$ from the above system to the system $X$
  given by
  \[
    \to
    x \xrightarrow{~~0.5~~}
    y \xrightarrow{~~0.7~~}
    z.
  \]
  While $A$ has no run for the path $\bullet \xrightarrow{0.5} \bullet$, the
  system $X$ has. So $f\colon A\to X$ is not an open map, even though it is a
  functional bisimulation, i.e.~coalgebra homomorphism for the subdistribution
  functor.
\end{enumerate}


\end{appendix}
\end{document}

